%% file: revision_or.tex
\documentclass[11pt,notitlepage]{article}
\usepackage{amsmath,amsthm,amssymb}
\usepackage[dvips]{graphicx}
\usepackage{array}
\usepackage[mathscr]{eucal}
\usepackage{eurosym}
\usepackage{subfigure}
\usepackage{color}
\usepackage{listings}
\usepackage{multirow} 
\usepackage{titling}
\usepackage{comment}

\usepackage[ruled,vlined]{algorithm2e}
\usepackage{verbatim}

\usepackage{tikz}
\usetikzlibrary{patterns}
\usepackage[outline]{contour}
\contourlength{1.2pt}

\usepackage{hyperref}

\usepackage{natbib}

\newtheorem{Theorem}{Theorem}[section]
\newtheorem{lemma}[Theorem]{Lemma}
\newtheorem{proposition}[Theorem]{Proposition}

\theoremstyle{definition}
\newtheorem{definition}[Theorem]{Definition}
\newtheorem{example}[Theorem]{Example}
\newtheorem{hypothesis}[Theorem]{Hypothesis}
\theoremstyle{remark}

\numberwithin{equation}{section}
\numberwithin{figure}{section}
\numberwithin{table}{section}

\newcommand{\bi}{ \begin{itemize}  }
\newcommand{\ei}{\end{itemize}}

\newcommand{\bfor}{ \begin{eqnarray*} }
\newcommand{\efor}{\end{eqnarray*}}

\begin{document}

\title{Collaborative Insurance Sustainability\\and Network Structure\\}
\author{Arthur Charpentier$^{a*}$, Lariosse Kouakou$^{b}$, Matthias L\"owe$^{c}$, \\Philipp Ratz$^{a}$ \& Franck Vermet$^{b}$ \\\\
        \small $^{a}$ Université du Québec à Montréal (UQAM), Montréal (Québec), Canada \\
        \small $^{b}$ EURo Institut d'Actuariat (EURIA), Université de Brest, France\\
        \small $^{c}$ University of Münster, Germany \\\\
        \small $^{*}$Corresponding author \tt{charpentier.arthur@uqam.ca}}
\begin{titlingpage}
\maketitle

\begin{abstract}
The peer-to-peer (P2P) economy has been growing with the advent of the Internet, with well known brands such as {\em Uber}, {\em Airbnb} or {\em Prosper} being examples thereof. Whereas the P2P lending sector has already received considerable attention in both practice and literature, the insurance is still in its infancy. Companies such as {\em Lemonade} in the U.S. or {\em Inspeer} in France have recently started to explore P2P-based collaborative insurance products and the academic literature has also begun to study such products from a mathematical perspective. Previous research in P2P lending motivates us to consider P2P insurance not merely as a different form of traditional insurance, but instead as a complement thereof. Accordingly, we propose a novel toolkit to analyse P2P insurance products based on the explicit structure of (P2P) networks rather than individual characteristics. We study the general case of risk sharing via network connections using convex order and extend the framework to incorporate characteristics typically encountered in P2P settings via simulations. 
	\\[1mm]
		\\[1mm]
\noindent{\bf Acknowledgement}: The authors wish to thank Romuald \'Elie, Mathieu Laurière, Tran Viet Chi, Miguel Campista, Luis Costa, Matteo Sammarco, Thomas Richard and Harpreet Kang as well as participants of the IME 2021 Conference, the Quantact ``summer seminar", the ASTIN 2021 Colloquium and the attendees of the Annual Meeting of the Statistical Society of Canada (SSC) 2022 for discussions and feedbacks. Arthur Charpentier received financial support from the Natural Sciences and Engineering Research Council of Canada (NSERC-2019-07077) and the AXA Research Fund.
Matthias L\"owe's research was funded by the Deutsche Forschungsgemeinschaft (DFG, German Research Foundation) under Germany's Excellence Strategy EXC 2044-390685587, Mathematics M{\"u}nster: Dynamics - Geometry - Structure.
\end{abstract}
\end{titlingpage}

% \tableofcontents
% \newpage

%%%%%%%%%%%%%%%%%%%%%%%%%%%%%%%%%%%%%%%%%%%%%%%%%%%%%%%%%%%%%%%%%%%%%%%%
\section{Introduction}

% {\color{red}{Main changes:
% \begin{enumerate}
%     \item Made extensive links to P2P lending, as this is already well studied in the OR literature. This gives some justification to the homophily property that we use to justify i.i.d. losses. Additionally, it allows us to present our framework as a "fusion" of what was done by eg. \cite{abdikerimova2022peer} and what was already researched in the P2P lending literature, so we do not have to reinvent the wheel. 
%     \item Removed deductible part from linear risk sharing, this way we can ensure that the results derived using the convex order hold - instead it is moved down to the simulations. This had the benefit that we can present the simulations as a natural extension of the first (analytical) part, but with a focus on the variance. Further, it would allow to find an optimum if the we have a quadratic utility function of the agents. Maybe we can extend this to show that if we use neither linear nor quadratic utilities the optimization problem becomes too complicated (eg. the linear program for individual utilities will become NP-complete)
%     \item Changed focus of friends-of-friends and made link to P2P lending. 
%     \item Moved story on cliques and a few examples to the appendix (will be online) - otherwise we will exceed the maximum pagecount
%     \item Still a todo: Need to have a short discussion about "nonlinear" risk sharing, is it really non-linear?
% \end{enumerate}

% ToDo:
% \begin{itemize}
%     \item Write appendix
% \end{itemize}
% }}

The peer-to-peer (P2P) economy has been steadily growing with the advent of the internet, {\em Uber} (ride services), {\em Airbnb} (lodging) or {\em Prosper} (lending) being prominent examples thereof. At its core, it establishes contracts between individuals for services that are either not available or otherwise deficient when acquired from traditional providers (taxi companies, hotels, and traditional credit institutions in the three cases above). A growing body of research has been studying the different business models and reasons why the P2P economy performs better in some tasks than the incumbent firms. In the literature concerning the financial sector, P2P {\em lending} has received considerable attention and its functional models have been studied extensively. Contrary to P2P lending, P2P {\em insurance} has only recently started to gain attention from an academic perspective, with \cite{guo2016instance}, \cite{hong2019optimal}, \cite{wang2022credit} \cite{abdikerimova2022peer} or \cite{Robert2} examples hereof. This article extends the research in the field by taking into account the specific nature of P2P financial products, for example the smaller scales that are typically encountered and analyses the consequences of these settings in detail. Instead of focusing on large-scale pooling that is at the core of traditional insurance, we focus on the opposite case, where smaller groups of peers work together to provide mutual insurance, similar to what companies such as {\em Friendsurance} propose. In that sense, our analysis is similar to situations that were already studied in the context of P2P lending and we are able to draw from the insights of previous research. Furthermore, instead of focusing on a fixed pool of individuals, we consider a framework that is based explicitly on network ties, in the spirit of previous P2P research.

In what follows we will first focus on the similarities between P2P lending and P2P insurance, which can be summarised in three points as described below. (1) Both products contain a partly unforeseeable risk which may result in a loss. For lending this corresponds to a default for insurance to a loss in the actuarial sense. (2) Trust between the parties is of essence, in lending this relates to securing funding (see for example \cite{thakor2018trust})
in insurance this boils down to having a counterparty that acts in good faith (the nature of this will be discussed further below). (3) Adverse selection may be present in both markets under traditional settings.

Similarity (1) is inherent in most financial applications, and typically interest or a premium is charged to compensate the issuer of a security or an insurance contract for the risk. It is common practice for traditional creditors and insurers to screen the applicants before entering a contract, either with some form of a credit score or a risk model based on available "hard" information, such as socio-demographics and credit history. The issue here is that the information is inherently incomplete which then often results in problems with similarity (3), as only an individual knows its true risk and will try to optimize accordingly. In P2P settings this information is often not available but instead "soft" information, that is found in social behaviour or information, is used by the peers to assess a given risk. \cite{iyer2016screening} find that risk evaluations using such indicators can even outperform predictions done by the exact credit score. Within a given social network, as \cite{granovetter2005impact} argues, social ties affect the flow of information between connected agents (which we will refer to as "friends"), in such a way that the soft information described above is easier to transmit between them. Beyond the fact that friends usually have access to better soft information, \cite{lin2013judging} find that online friendships also act as a signal of credit quality. This can be partly explained by the property of network homophily (cf. \cite{McPherson_Smith-Lovin_Cook_2001}) which states that within a network, similar agents are more likely attached to one-another than dissimilar ones. If agents are more similar, this is also an indicator that they possess similar risk profiles.

Similarity (2) actually has two dimensions in the case of insurance. From the insurers side, this means that individuals might engage in risky behaviour once they are covered but also entails moral hazard by for example including invalid claims in bad faith, as discussed in \cite{tennyson2008moral}. %{\color{red} ToDo: find some literature}. 
From the policyholders side the trust issue condenses to believing that the insurer actually reimburses claims. As \cite{abdikerimova2022peer} write based on  \cite{carlin2009strategic}, insurers have an incentive to increase their profit by introducing exemptions that can be complicated for the policyholder to understand. If there is a social relationship between both the insurer and the policyholder, for example if they are mutually insured, we can employ the concept of social collateral. The profit of gaining through the refusal of a valid claim then receives a negative effect on the social plane. \cite{liu2020social} use the concept to show that the danger of losing friendships or other social ties can also be effective for financial aspects, for example it can be used to guarantee the payback of a loan similar to traditional collateral. \cite{karlan2009trust} argue that networks generate various degree of trust between the agents, a concept which we will also exploit later on.  In the case of insurance, social collateral can apply to both directions of similarity (2). Further, moral hazard might be actively reduced in such setting, as for example \cite{biener2018can} and \cite{benjaafar2019peer} find that pro-social preferences that are often present in social networks can alleviate such issues.

Finally, similarity (3) is often an issue affiliated with information asymmetry. Whereas collateral is often demanded in the lending case, deductibles are frequently employed in the insurance industry. In both cases this can lead to a sub-optimal outcome, we will discuss the implications of this later on. Note though that in traditional setting, information asymmetry arises because the risk can only be approximated by an institution, hence the arguments from similarity (1) should also apply to this point. Although much common ground exists between P2P lending and insurance, there are also points that differentiate them. Most notably, insurance is based on the concept that ``{\em the contribution of the many to the misfortune of the few}'' as written in \cite{charpentierdenuit}. Whereas lending usually has an aspect that is similar to an investment, insurance can be mutually beneficial even in a zero sum game, as mutualization reduces the variance of risks which benefits risk-averse agents. Further, insurance often possesses a social component, that is, agents are willing to pay more than their (actuarially) fair share for some common social objective. That is a feature in many social insurance systems that go beyond the pure expectation of costs, as an example one might consider the Swiss health insurance system \cite{schindler2018solidarity}.

We conclude the introduction by analysing the existing research in P2P insurance and point out how our approach complements the studies in the field. Recently \cite{abdikerimova2022peer} investigated mutual aid, a P2P-insurance based approach popular in China. They formalize and propose different risk-sharing mechanisms and derive analytical solutions thereof. \cite{feng_2020_flood} study decentralized insurance schemes and propose an allocation mechanism that is Pareto optimal. Our approach here differs from these two recent articles in the sense that we consider heterogeneous risks stemming from a network of peers rather than the peers themselves. Our analysis focuses on arbitrary networks and nests the types considered in the articles above as special cases (more on that later). Another direction of research can be found for example in \cite{Robert} and \cite{Robert2}, who mainly study asymptotic behaviour of P2P insurance schemes when they grow large. Again, our approach complements this research as we follow the findings of P2P lending and will assume that the insurance contribution of peers are typically small and stem from a small group of peers. In short, we propose a mechanism that allows to compare different P2P insurance schemes based on arbitrary (social) networks and analyse the conditions for a sustainable system with small pools rather than arbitrarily large ones. In that we also follow the P2P lending literature such as \cite{tang2019peer} who argue P2P solutions are a complement rather than a substitute to traditional solutions. Examples of risk sharing in industrial settings can for example be found in \cite{buzacott2012contract}. The remainder of the article is structured as follows: After briefly recalling the most important aspects of graph theory that will be used in this article in Section \ref{sec:intrographs}, we propose a toolset to analyse and order different network-based P2P insurance schemes based on their convex order in Section \ref{sec:linearrisksharing}. We then extend the findings to nonlinear risk sharing with the use of simulations in \ref{sec:nonlinear}, this in turn allows us to incorporate the results into a utility framework often used in insurance. We discuss some brief extensions and conclude with Section \ref{sec:conclusion}.

\section{Notations from Graph Theory}\label{sec:intrographs}

As this article will deal with arbitrary networks, we need to generalize a single network into a {\em graph}. This section recalls the most important notations from graph theory that will be used throughout the subsequent sections. A graph $\mathcal{G}$ on $\{1,2,\dots,n\}$ is a pair $(\mathcal{V}, \mathcal{E})$, where $\mathcal{V}=\{1,2,\dots,n\}$ is the set of vertices (also called nodes) that correspond to the insurees and $\mathcal{E}$ is a set of edges (also called links), which are pairs of vertices, e.g.\ $\{i,j\}$, $i,j\in \mathcal{V}$. Here we will consider undirected networks, where an edge $\{i,j\}$
means that vertices $i$ and $j$ have a \emph{reciprocal} connection. To make the analysis tractable, we assume that the network is static for the observed time period such that neither nodes nor edges can added.

The graph can then be represented by an $n\times n$ matrix $A$, called adjacency matrix, where $A_{i,j}\in\{0,1\}$ with $A_{i,j}=1$ if and only if $\{i,j\}\in \mathcal E$. Let $A_{i\cdot}$ (and $A_{\cdot i}$) denote the vector in $\{0,1\}^n$ that corresponds to row $i$ (and column $i$) of matrix $A$. As we will show, for the purpose of P2P insurance the crucial component within the network structure is the node \emph{degree} which corresponds to the cardinality of the set of nodes connected to it. Given a network $\mathcal{G}=(\mathcal{V}, \mathcal{E})$ and the associated adjacency matrix $A$, the degree of a node $i$ is
$$
d_i=\displaystyle{A_{i\cdot}^\top \boldsymbol{1}=\sum_{j=1}^n A_{i,j}}
$$
For convenience, let $\mathcal{V}_i$ denote connections of node $i$, i.e. $\mathcal{V}_i=\lbrace j\in\mathcal{V}:\{i,j\}\in\mathcal{E}\rbrace = \lbrace j \in V:  A_{i,j}=1\rbrace$. The \emph{degree vector} $\boldsymbol{d}$ summarizes the degree information across the network such that $\boldsymbol{d}=[|\mathcal{V}_1|,\dots, |\mathcal{V}_n|]$. Our focus in this article will be random networks that are specified by the expected value of the degree vector $\mathbb{E}[\boldsymbol{d}]$ (corresponding to the average, $\displaystyle{\overline{d}=\boldsymbol{d}^\top \boldsymbol{1}/n}$) and its variance $\text{Var}[\boldsymbol{d}]$ (corresponding to $\displaystyle{(\boldsymbol{d}-\mathbb{E}[\boldsymbol{d}]\boldsymbol{1})^\top \boldsymbol{1}/n}$))\footnote{Note focusing on the first two moments makes the analysis merely more tractable for the cases we consider here, extensions could also include networks that are specified otherwise.}. The precise mechanism that is used to generate the networks will be discussed in Section \ref{sec:nonlinear}.

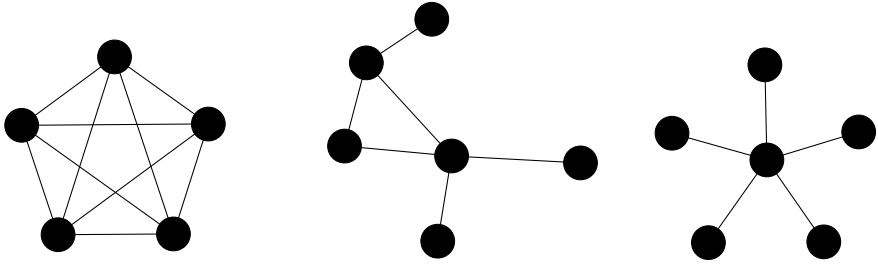
\begin{figure}
    \centering
    \input{tikz/graph_illustration.tex}
    \caption{Left, a complete and therefore regular graph where $\text{Var}[d] = 0$, Center an example of a general graph with $\text{Var}[d] \neq 0$ and right, the centralized model. The illustrations on the flanks are as described in \cite{abdikerimova2022peer}, the general case in the center will be discussed in the article.}
    \label{tikz:graphrep}
\end{figure}

{%\color{red}
The situation when the individual degrees are identically distributed and $\text{Var}[\boldsymbol{d}]=0$ merits special consideration. In this case we are dealing with a \emph{regular graph}, that is a network where every node $i$ has exactly the same amount of nodes adjacent to it. If in addition $\boldsymbol{d}=(n-1)$ for all $i$, the resulting graph will be \emph{complete}, in this case every node within $\mathcal{G}$ is connected to one another. The traditional centralized insurance model can also be considered a special kind of network. In this case, the insurer acts as pooling instance that helps to connect the nodes indirectly and allows them to share risk among each other for the benefit of all. Such a collective transfer is possible because of the central limit theorem that guarantees that the relative variability of individual contracts decreases with the number of insured in the pool. The P2P alternative to traditional insurance often works in a similar way, instead of having an insurer in the center the often work with a common fund that acts as a pooling mechanism. The articles of \cite{abdikerimova2022peer} and \cite{feng_2020_flood} consider special cases thereof, where either the entire graph is complete or the graph can be broken down into large cliques (complete subgraphs). In this article however, this will generally not be the case. Figure \ref{tikz:graphrep} illustrates the different approaches, we will consider general graphs like the center image, which nests the two graphs on the side as special cases\footnote{An example that combines their approach and ours can be found in the online supplementary material}.}

\section{Linear Risk Sharing and Convex Order}\label{sec:linearrisksharing}

In order to obtain some general results on the sharing schemes and the impact of the shape of the network, it is necessary to make some assumptions which we discuss here. Consider $n$ nodes, indexed with $i=1,\dots,n$. Each node corresponds to a policyholder, facing a random yearly loss, that might occur with probability $p$, and cost $Y_i$, that has distribution $F$. For convenience, let $Z_i$ denote the indicator for the event that $i$ claims a loss, as it is common in the individual model in actuarial science. $Z_i$ is a Bernoulli variable with parameter $p$. Thus, we consider the simple case where only one loss might occur. The individual risk is then defined as $X_i = Z_i\cdot Y_i$. For the variable $Z_i$ we work with the following hypothesis:

\begin{hypothesis}[Loss frequency]\label{hyp90}
Policyholders can face only one loss, or none, over a year.
\end{hypothesis}

We mainly assume Hypothesis \ref{hyp90} to facilitate the model and to describe the mechanism. It is also possible to assume that many claims might occur, and that the total sum of contributions is bounded (by some value $\gamma$). This would not substantially change the analysis, but will make notations more cumbersome. Based on the concept of network homophily described above, we also introduce the following hypothesis:

\begin{hypothesis}[Risk distribution]\label{hyp9}
All risks $X_i$'s are assumed to be independent and identically distributed random variables.
\end{hypothesis}

Hypothesis \ref{hyp9} is important to have a simple reciprocal and fair mechanism. If two policyholders $i$ and $j$ do not have the same probability to claim a loss (distribution of $Z$), or the same distribution of individual claims (distribution of $Y$), it becomes more difficult to derive a fair contribution for some risk exchange (this was discussed recently in \cite{Robert}). The hypothesis can then be relaxed in Section \ref{sec:nonlinear}. In what follows, we first introduce the concept of convex order to compare different insurance mechanisms in general and then propose linear risk sharing for the P2P problem at hand specifically.

\subsection{Risk preferences and ordering of risk}

\cite{Ohlin1969}, inspired by \cite{Karlin1963}, suggested to use the convex order to solve the optimal insurance decision problem, as discussed in \cite{denuit2006actuarial}, \cite{egozcue2010gains} or \cite{CHEUNG2015409}. \cite{grechuk2015center} used it in the context of capital allocation (which satisfies properties close to risk sharing ones). Given a random loss $X$, and an insurance scheme offering indemnity $x\mapsto I(x)$ against a premium $p$, an agent will purchase the insurance if $\xi\preceq_{CX}X$ for the convex order, where $\xi=X+\pi-I(X)$ (where $\pi$ is the associated premium to transfer that risk). This formalism is a natural translation of the expected utility model in insurance, where an agent having wealth $w$ and (concave) utility $u$ agrees to purchase a contract offering utility $I$ against a premium $\pi$ if $\mathbb{E}[u(w-X)]\leq \mathbb{E}[u(w+I(X)-\pi-X)]=\mathbb{E}[u(w-\xi)]$. This framework will also allow us to study the desirability of different P2P insurance mechanisms. Here we follow \cite{shaked2007stochastic} or \cite{denuit2006actuarial}, and define the convex order as follows:

\begin{definition}[Convex order]
Consider two variables $X$ and $Y$ such that $$\mathbb{E}[g(X)]\leq \mathbb{E}[g(Y)]\text{ for all convex functions }g:\mathbb{R}\rightarrow \mathbb{R},$$
(provided the expectations exist). Then $X$ is said to be smaller than $Y$ in the convex order, and is denoted as $X\preceq_{CX}Y$.
\end{definition}\label{def:convexorder}
As we will see below,
the inequality $X\preceq_{CX}Y$ intuitively means that $X$ and $Y$ have the same magnitude, as $\mathbb{E}[X] =\mathbb{E}[Y]$, but that $Y$ is more variable than $X$. 
More specifically:

\begin{proposition}\label{prop-mean-var}
If $X\preceq_{CX}Y$, then $\mathbb{E}[X] =\mathbb{E}[Y]$ and $\text{Var}[X] \leq\text{Var}[Y]$.
\end{proposition}
\begin{proof}
For the expected value, consider $g(x)=\pm x$ (Equation (3.A.2) in \cite{shaked2007stochastic}) and for the variance $g(x)=x^2$ (Equation (3.A.4) in \cite{shaked2007stochastic}).
\end{proof}

If $X\preceq_{CX}Y$, we can also say that $Y$ is a {\em mean preserving spread} of $X$, in the sense that it satisfies a martingale property, that can be written $Y\overset{\mathcal{L}}{=}X+Z$ where $Z$ is such that $\mathbb{E}[Z|X]=0$, as discussed in \cite{CARLIER2012207}.
An example of variables ordered via the convex order is given in the following lemma:

\begin{lemma}\label{lemma-prob-cx}
Let $\boldsymbol{X}=(X_1,\cdots,X_n)$ denote a collection of i.i.d. variables, and $\boldsymbol{p}$ some $n$-dimensional probability vector. Then $\boldsymbol{p}^\top\boldsymbol{X}\preceq_{CX}X_i$ for any $i$.
\end{lemma}

\begin{proof}
Let $\boldsymbol{X}^+_i=X_i\boldsymbol{1}$, which is a commonotonic version of vector $\boldsymbol{X}$, in the sense that all marginals have the same distribution (since components of $\boldsymbol{X}$ are identically distributed), and $X_i=\boldsymbol{p}^\top\boldsymbol{X}^+_i$. A simple extension of Theorem 6 in \cite{kaas2002} (and Proposition 3.4.29 in \cite{order}) implies that $\boldsymbol{p}^\top\boldsymbol{X}\preceq_{CX}\boldsymbol{p}^\top\boldsymbol{X}^+_i$ for any $i$ (since $X_i$'s are i.i.d.), i.e. $\boldsymbol{p}^\top\boldsymbol{X}\preceq_{CX}X_i$. Note that it can also be seen as a corollary of Property 3.4.48 in \cite{order}, using the fact that $\boldsymbol{p}\prec \boldsymbol{e}_i$ for the majorization order ({%\color{red} 
where $\boldsymbol{e}_i$ is the $i$th standard basis vector, with $1$ at the $i$th position and $0$ anywhere else, i.e. $\boldsymbol{e}_i=(0,0,\cdots,0,1,0,\cdots,0,0)^\top$)}.
\end{proof}

If such a concept can be used to compare two risks, for a single individual, we need a multivariate extension to study the risk over the entire portfolio. Following \cite{DENUIT2012265}, define the componentwise convex order as follows:

\begin{definition}[Componentwise convex order]
Given two vectors $\boldsymbol{X}$ and $\boldsymbol{Y}$, $\boldsymbol{X}$ is said to be smaller than $\boldsymbol{Y}$ in the componentwise convex order if $X_i\preceq_{CX}Y_i$ for all $i=1,2,\dots,n$, and is denoted $\boldsymbol{X}\preceq_{CCX}\boldsymbol{Y}$. If at least one of the convex order inequalities is strict, we will say that  $\boldsymbol{Y}$ dominated $\boldsymbol{X}$, and is denoted $\boldsymbol{X}\prec_{CCX}\boldsymbol{Y}$.
\end{definition}\label{def3-6}

This inequality $\boldsymbol{X}\prec_{CCX}\boldsymbol{Y}$ means that, for a pool of $n$ insured, each and everyone of them prefers $X_i$ to $Y_i$. As mentioned in the introduction, some social insurance systems instead consider a sort of {\em global} objective to instead maximize total welfare (as an example we can again take the Swiss health insurance system \cite{schindler2018solidarity}). Thus it should also be possible to define an order which considers the global interest instead of the componentwise order.  This means that we need to be slightly less restrictive in order to study some possible collective gain of such a coverage, for the group of $n$ insured. In what follows, we can consider some weaker conditions. The first one will be related to the classical idea in economics of a {\em representative agent}, but here, since there might be heterogeneity of agents, we will consider some {\em randomly chosen agent}. A weaker condition than the one in Proposition \ref{prop-mean-var} would be: let $I$ denote a random variable, uniformly distributed over $\{1,2,\cdots,n\}$, and define $\xi'=\xi_I$, then we wish to have
$$
\mathbb{E}[\xi']=\mathbb{E}\big[\mathbb{E}[\xi_I|I]\big]=\mathbb{E}[X_i]\text{ and }\text{Var}[\xi']=\mathbb{E}\big[\text{Var}[\xi_I|I]\big]=\frac{1}{n}\text{trace}\big[\text{Var}[\boldsymbol{\xi}]\big]\leq \text{Var}[X_i]=\sigma^2,
$$
and not necessarily
$$
\mathbb{E}[\xi_i]=\mathbb{E}[X_i]\text{ and }\text{Var}[\xi_i]\leq \text{Var}[X_i],~\forall i\in\{1,2,\dots,n\}.
$$
Note that \cite{Chatfield1981} %and \cite{Johnson2008} 
call $\text{trace}\big[\text{Var}[\boldsymbol{\xi}]\big]$ the {\em total variation} of vector $\boldsymbol{\xi}$. We will not use that name here, as such a quantity is a standard variability measure for random vectors.

\subsection{Linear Risk Sharing}

{%\color{red} 
Loosely stated, a risk sharing scheme of a portfolio of risks is a random vector that reallocates the total risk of the portfolio among the agents. A special form which we will use here is the {\em linear risk sharing scheme}. The use of this in turn allows us to give a tractable solution to P2P insurance on arbitrary networks}. We first follow the definition of \cite{DENUIT2012265} and define an arbitrary risk sharing scheme as:

\begin{definition}[Risk sharing scheme]\label{def:risksharing}
Consider two random vectors $\boldsymbol{\xi}=(\xi_1,\dots,\xi_n)$ and  $\boldsymbol{X}=(X_1,\cdots,X_n)$ on $\mathbb{R}_+^n$. Then $\boldsymbol{\xi}$ is a risk-sharing scheme of $\boldsymbol{X}$ if $X_1+\cdots+X_n=\xi_1+\cdots+\xi_n$ almost surely.
\end{definition}

As an example, the average is a risk sharing principle: $\xi_i=\overline{X}=\displaystyle{\frac{1}{n}\sum_{j=1}^n}X_j$, for any $i$. The definition of linear risk sharing scheme then follows directly:

\begin{definition}[Linear Risk sharing scheme]\label{def:linearrisksharing}
Consider two random vectors $\boldsymbol{\xi}=(\xi_1,\dots,\xi_n)$ and  $\boldsymbol{X}=(X_1,\dots,X_n)$ on $\mathbb{R}_+^n$, such that $\boldsymbol{\xi}$ is a risk-sharing scheme of $\boldsymbol{X}$. It is said to be a linear risk sharing scheme if there exists a matrix $M$, $n\times n$, with positive entries, such that $\boldsymbol{\xi}=M\boldsymbol{X}$, almost surely.
\end{definition}

As an example, consider the average $\xi_i=\overline{X}$ for $i=1,2,\cdots,n$, characterized by matrix $\bar M=[\bar M_{i,j}]$ with $\bar M_{i,j}=1/n$. A direct extension can be obtained with any doubly stochastic matrix. A $n\times n$ matrix $D$ is said to be doubly stochastic if
$$
D=[D_{i,j}]\text{ where }D_{i,j}\geq0,~D_{\cdot,j}^\top \boldsymbol{1}=\sum_{i=1}^n D_{i,j}=1~\forall j,~\text{and}~D_{i,\cdot}^\top \boldsymbol{1}=\sum_{j=1}^n D_{i,j}=1~\forall i.
$$

Then $\boldsymbol{\xi}=D\boldsymbol{X}$ is a linear risk sharing of $\boldsymbol{X}$. To make the link to the literature, row and column conditions on matrix $D$ are called {\em zero-balance conservation} in \cite{feng_2020_flood}. A particular case of a doubly stochastic matrix is a permutation matrix D associated with a permutation of $\lbrace 1, \dots, n \rbrace$. This example with doubly stochastic matrices is important because of connections with the majorization concept (see \cite{schur} and \cite{HLP}). Furthermore \cite{Chang1992}, \cite{Chang93} extended the deterministic version to some stochastic majorization concept, with results that we will use here.

Inspired by the componentwise convex ordering, the following definitions can be considered, as in \cite{DENUIT2012265} (a similar concept can be found in \cite{CARLIER2012207}) that allow us to compare a risk sharing scheme to the no-insurance case

\begin{definition}[Desirable risk sharing schemes (1)]
A risk sharing scheme $\boldsymbol{\xi}$ of $\boldsymbol{X}$ is desirable if $
\boldsymbol{\xi}\preceq_{CCX}\boldsymbol{X}$.
\end{definition}

and to compare risk sharing schemes among each other:

\begin{definition}[Ordering of risk sharing schemes (1)]
Consider two risk sharing schemes $\boldsymbol{\xi}_1$ and $\boldsymbol{\xi}_2$ of $\boldsymbol{X}$. $\boldsymbol{\xi}_1$ dominates $\boldsymbol{\xi}_2$, for the convex order, if $
\boldsymbol{\xi}_{2}\preceq_{CCX}\boldsymbol{\xi}_{1}$.
\end{definition}

As an example to differentiate between a risk-sharing scheme and a desirable risk sharing scheme consider the order statistic. Defined as $\xi_i=X_{(i)}$ where $X_{(1)}\leq X_{(2)}\leq \cdots \leq X_{(n)}$. Given a permutation $\sigma$ of $\{1,2,\dots,n\}$, $\xi_i=X_{\sigma(i)}$, the order statistic defines a risk sharing scheme but not a desirable one. This can be easily seen if the $X_i$ are non-deterministic ordered values since $\mathbb{E}[\xi_n]>\mathbb{E}[X_n]$ we cannot have $\xi_n\preceq_{CX}X_n$. We then extend this framework to the linear risk sharing scheme with the following proposition:

\begin{proposition}[Desirable linear risk sharing schemes (1)]\label{prop:3:10}
If $\boldsymbol{\xi}$ is a linear risk sharing of $\boldsymbol{X}$, $
\boldsymbol{\xi}=D\boldsymbol{X}$ for some doubly stochastic matrix $D$, then $\boldsymbol{\xi}$ is desirable.
\end{proposition}
\begin{proof}
Assume that $
\boldsymbol{\xi}=D\boldsymbol{X}$ for some double stochastic matrix $D$, with rows $D_{i\cdot}$'s. Since $D$ is doubly stochastic, $D_{i\cdot}$ is a probability measure over $\{1,\dots,n\}$ (with positive components that sum to one), and from Lemma \ref{lemma-prob-cx}, $D_{i\cdot}^\top \boldsymbol{X}\preceq_{CX} X_i$, for all $i$.
\end{proof}

We then extend the proposition to the ordering of linear risk sharing mechanisms. This will allow us to compare different linear risk sharing schemes:

\begin{proposition}[Ordering of Linear risk sharing schemes (1)]\label{prop3-12}
Consider two linear risk sharing schemes $\boldsymbol{\xi}_1$ and $\boldsymbol{\xi}_2$ of $\boldsymbol{X}$, such that there is a doubly stochastic matrix $D$, $n\times n$ such that $
\boldsymbol{\xi}_{2}=D\boldsymbol{\xi}_{1}$. Then $
\boldsymbol{\xi}_{2}\preceq_{CCX}\boldsymbol{\xi}_{1}$.
\end{proposition}

\begin{proof}
This can be seen as an extension of Proposition \ref{prop:3:10}.
Here $\boldsymbol{\xi}_{1}=M_1\boldsymbol{X}$ and $
\boldsymbol{\xi}_{2}=M_2\boldsymbol{X}$, since we consider linear risk sharings, and we can write the later $\boldsymbol{\xi}_{2}=M_2\boldsymbol{X}=DM_1\boldsymbol{X}$, thus $M_2=DM_1$, since equality is valid for all $\boldsymbol{X}$. As defined in \cite{DAHL199953} and \cite{BEASLEY2000141}, this corresponds to the {\em matrix majorization} concept from above, in the sense that $M_2\prec M_1$, that implies standard majorization per row, thus for all $i$, $M_{2:i\cdot}\prec M_{1:i\cdot}$, from Lemma 2.8 in \cite{BEASLEY2000141}. 
But from Property 3.4.48 in \cite{order}, we know that if $\boldsymbol{a}\prec\boldsymbol{b}$ for the majorization order, $\boldsymbol{a}^\top\boldsymbol{X}\preceq_{CX}\boldsymbol{b}^\top\boldsymbol{X}$. Thus, if $\boldsymbol{a}=M_{2:i\cdot}$ and $\boldsymbol{b}=M_{1:i\cdot}$, we have that for any $i$, $\xi_{2:i}\preceq_{CX}\xi_{1:i}$, and therefore we have the componentwise order $
\boldsymbol{\xi}_{2}\preceq_{CCX}\boldsymbol{\xi}_{1}$.
\end{proof}

As we will show in Section \ref{sec:nonlinear}, the variance of the degree vector will have a major influence on the sustainablity of an insurance scheme. Hence we consider the following proposition:

\begin{proposition}[Variance comparison]
Consider two linear risk sharing schemes $\boldsymbol{\xi}_1$ and $\boldsymbol{\xi}_2$ of $\boldsymbol{X}$, such that $\boldsymbol{\xi}_2=D\boldsymbol{\xi}_1$, for some doubly-stochastic matrix. Then
$\operatorname{Var}[\boldsymbol{\xi}_{2:i}]\leq\operatorname{Var}[\boldsymbol{\xi}_{1:i}]$ for all $i$.
\end{proposition}

\begin{proof}
This is simply the fact that from Proposition \ref{prop3-12}, $
\boldsymbol{\xi}_{2}\preceq_{CCX}\boldsymbol{\xi}_{1}$, i.e. from Definition \ref{def3-6}, $\boldsymbol{\xi}_{2:i}\preceq_{CCX}\boldsymbol{\xi}_{1:i}$ for all $i$ and $\operatorname{Var}[\boldsymbol{\xi}_{2:i}]\leq\operatorname{Var}[\boldsymbol{\xi}_{1:i}]$ follows from the property on the convex order mentioned earlier.
\end{proof}

A simple corollary is the following:

\begin{proposition}[Trace of variance matrix]\label{prop:3:7}
Consider two linear risk sharing schemes $\boldsymbol{\xi}_1$ and $\boldsymbol{\xi}_2$ of $\boldsymbol{X}$, such that $\boldsymbol{\xi}_2=D\boldsymbol{\xi}_1$, for some doubly-stochastic matrix. Then
$\operatorname{trace}\big[\operatorname{Var}[\boldsymbol{\xi}_2]\big]\leq\operatorname{trace}\big[\operatorname{Var}[\boldsymbol{\xi}_1]\big]$.
\end{proposition}

\begin{proof}
Observe first that $\operatorname{Var}[\boldsymbol{\xi}_2]=D\operatorname{Var}[\boldsymbol{\xi}_1]D^\top$. Since the variance matrix $\operatorname{Var}[\boldsymbol{\xi}_1]$ is positive, the application $h:M\mapsto \operatorname{trace}\big[M\operatorname{Var}[\boldsymbol{\xi}_1]M^\top\big]$ is convex. By the Birkhoff-von Neumann theorem (see \cite{birkhoff1946three}), every doubly stochastic matrix is a convex combination of permutation matrices, i.e. $D=\omega_1 P_1+\cdots+\omega_k P_k$ for some permutation matrices $P_1,\cdots,P_k$ (and some positive weights $\omega_i$ that sum to 1). By the convexity of $h$,
$$
\operatorname{trace}\big[\operatorname{Var}[\boldsymbol{\xi}_2]\big]=h(D)\leq \sum_{i=1}^k\omega_i h(P_i)=\sum_{i=1}^k\omega_i \operatorname{trace}\big[P_i\operatorname{Var}[\boldsymbol{\xi}_1]P_i^\top\big],
$$
and therefore
$$
\operatorname{trace}\big[\operatorname{Var}[\boldsymbol{\xi}_2]\big]\leq\sum_{i=1}^k\omega_i \operatorname{trace}\big[\operatorname{Var}[\boldsymbol{\xi}_1]\big]= \operatorname{trace}\big[\operatorname{Var}[\boldsymbol{\xi}_1]\big].
$$
\end{proof}

A special case of the linear risk sharing arises when the graph is complete. In this case the minimum, will arise when every node in the graph assumes $\frac{1}{n}$ of the total risk. A closely related example arises when subgroups are complete, as for example in \cite{feng_2020_flood}. A brief consideration of such a situation can be found in the supplementary materials.

\subsection{Weaker orderings}

{%\color{red}
P2P insurance might contain a strong {\em social} component, where the focus is on global welfare rather than individual well-being. See for example the case of \emph{takaful} in \cite{feng2022unified} or \emph{friend} in friendsurance}. In order to study a situation that goes beyond individual optimization, we also consider weaker orderings of insurance mechanisms. To this end, we follow \cite{MARTINEZPERIA2005343}
and replace the doubly stochastic matrix in the previous section  by a column stochastic matrix (or left-stochastic matrix). By definition
such a matrix $C$ satisfies
$$
C=[C_{i,j}]\text{ where }C_{i,j}\geq0,~\text{and }\sum_{i=1}^n C_{i,j}=1~\forall j.
$$

Transforming the risk vector $\boldsymbol{X}$ with a column stochastic matrix then results in the following proposition:

\begin{proposition} Let $C$ be some $n\times n$ column-stochastic matrix, and given $\boldsymbol{X}$, a positive vector in $\mathbb{R}^+$, define $\boldsymbol{\xi}=C\boldsymbol{X}$. Then $\boldsymbol{\xi}$ is a linear risk sharing of $\boldsymbol{X}$.
\end{proposition}

\begin{proof}
First, observe that since $C$ is a matrix with positive entries, $\boldsymbol{\xi}\in\mathbb{R}^n_+$. Moreover, 
$$
\sum_{i=1}^n\xi_i=\sum_{i=1}^n \sum_{j=1}^n C_{i,j}X_j=\sum_{j=1}^n\left(\sum_{i=1}^n  C_{i,j}\right)X_j=\sum_{j=1}^nX_j. 
$$
\end{proof}

To illustrate why weaker orderings present a complement to the  orderings derived earlier consider the following example:

\begin{example}[Weaker ordering]\label{ex:simp}
     Consider the two following matrices (where $D$ corresponds to $C$ when $\alpha=0$),
    $$
    D = \begin{pmatrix}
    1 & 0 & 0 \\
    0 & 1/2 & 1/2 \\
    0 & 1/2 & 1/2 \\
    \end{pmatrix}
    \text{ and }
    C = \begin{pmatrix}
    1 & \alpha & \alpha \\
    0 & (1-\alpha)/2 & (1-\alpha)/2 \\
    0 & (1-\alpha)/2 & (1-\alpha)/2 \\
    \end{pmatrix},\text{ with }\alpha\in(0,1).
    $$
    Note that $D$ is doubly stochastic but $C$ is only column stochastic. Further, take a non-negative random vector  $\boldsymbol{X}=(X_1,X_2,X_3)$ where the three components are i.i.d.\ and represent the individual losses. Denote
    $\mu=\mathbb{E}[X_i]$ and $\sigma^2=\text{Var}[X_i]$.
    Consider first the risk sharing scheme $\boldsymbol{\xi} = (\xi_1, \xi_2, \xi_3)$,  $\boldsymbol{\xi} = D X$, that is $\xi_1=X_1$ and $\xi_2=\xi_3=\overline{X}_{23} = (X_2+X_3)/2$. We choose a random component of $\boldsymbol{\xi}$ selected by a random variable $I$ which is uniformly distributed over $\{1,2,3\}$ (and independent of $\boldsymbol{X}$). Then $\mu=\mathbb{E}[X_I]$, and $\sigma^2=\text{Var}[X_I]$. Thus
    $$
    \mathbb{E}[\xi_I]=\frac{1}{3}\sum_{i=1}^3\mathbb{E}[\xi_i]=\frac{1}{3}\left[\mu+2\cdot\mu\right]=\mu=\mathbb{E}[X_I],
    $$
    $$
    \text{Var}[\xi_I]=\frac{1}{3}\sum_{i=1}^3\text{Var}[\xi_i]=\frac{1}{3}\left[\sigma^2+2\cdot \frac{2\sigma^2}{4}\right]=\frac{2}{3}\sigma^2<\text{Var}[X_I].
    $$
    This is consistent with the fact that $\boldsymbol{\xi}$ is a desirable risk sharing principle (proposition \ref{prop:3:10}). Now consider a second mechanism $\boldsymbol{\eta} = (\eta_1, \eta_2, \eta_3)$ such that $\boldsymbol{\eta}=CX$. In that case we cannot compare $\boldsymbol{X}$ and $\boldsymbol{\eta}$ directly using proposition \ref{prop:3:10} since $C$ is not doubly stochastic. Using  $I$ in the same way as before, we get:
    $$
    \mathbb{E}[\eta_I]=\frac{1}{3}\sum_{i=1}^3\mathbb{E}[\eta_i]=\frac{1}{3}\left[(1+2\alpha)\mu+2\cdot(1-\alpha)\mu\right]=\mu=\mathbb{E}[X_I],
    $$
    $$
    \text{Var}[\eta_I]=\frac{1}{3}\sum_{i=1}^3\text{Var}[\eta_i]=\frac{1}{3}\left[(1+2\alpha^2)\sigma^2+\frac{4\cdot(1-\alpha)^2 \sigma^2}{4}\right]=\frac{1+2\alpha^2+(1-\alpha)^2}{3}\sigma^2<\text{Var}[X_I].
    $$
    So $CX$ and $DX$ have similar properties. However, note that this only works in expectation over $I$. In the second scheme, the first agent takes on more risk than she initially had, which does not suit well when the individual utility only depends on individual risk. Instead, if the utility function of that agent contains a social component, as in the social insurance case, this will still be rational. An example would be a family or friends who also value the global well-being of the group. \end{example}

Following \cite{MARTINEZPERIA2005343} further, we define an ordering based on column-stochastic matrices:

\begin{definition}[Weak Ordering of Linear risk sharing schemes (1)]
Consider two linear risk sharing schemes $\boldsymbol{\xi}_1$ and $\boldsymbol{\xi}_2$ of $\boldsymbol{X}$. $\boldsymbol{\xi}_1$ weakly dominates $\boldsymbol{\xi}_2$, denoted $
\boldsymbol{\xi}_{2}\preceq_{wCX}\boldsymbol{\xi}_{1}$ if and only if there is a column-stochastic $n\times n$ matrix $C$, such that 
$
\boldsymbol{\xi}_{2}=C\boldsymbol{\xi}_{1}$.
\end{definition}

In the example above, it is easy to see that $\boldsymbol{\eta}\preceq_{wCX}\boldsymbol{X}$ and with the definition we can compare weaker orderings. %, for example consider $C_2 = DC${\color {red} I  do not understand this equation}. 
Note though that Proposition \ref{prop:3:7} cannot be extended when matrices are only column-stochastic, in the sense that $
\boldsymbol{\xi}_{2}=C\boldsymbol{\xi}_{1}$ for some column-stochastic matrix $C$ is not sufficient to guarantee  $\text{trace}[\text{Var}[\boldsymbol{\xi}_2]]=\text{trace}[C\text{Var}[\boldsymbol{\xi}_1]C^\top]\leq \text{trace}[\text{Var}[\boldsymbol{\xi}_1]]$. Some heuristic interpretation of that result are given in the supplementary material online. However, such an (in)equality does hold when $\boldsymbol{\xi}_1$ are independent risks as shown in the following proposition:

\begin{proposition}
If $\boldsymbol{\xi}$ is a linear risk sharing of $\boldsymbol{X}$, where $X_i$'s are independent risks with variance $\sigma^2$, associated with some column-stochastic matrix $C$, then
$$
\operatorname{trace}\big[\text{Var}[\boldsymbol{\xi}]\big]=\sigma^2 \operatorname{trace}\big[CC^\top\big]\leq \sigma^2n
$$
and therefore $\operatorname{Var}[\xi']\leq\operatorname{Var}[X]=\sigma^2$.
\end{proposition}
 \begin{proof}
 In dimension $n$,
 $$
 \text{trace}\big[CC^\top\big]=\sum_{i,j=1}^n C_{i,j}^2 \leq \sum_{i,j=1}^n C_{i,j} = \sum_{i=1}^n \left(\sum_{j=1}^n C_{i,j} \right)=\sum_{i=1}^n 1=n.
 $$
 \end{proof}

Thus, the weaker ordering we introduced is only interesting if we introduce the ordering on the variance, that is:

\begin{definition}[Weak Ordering of Linear risk sharing schemes (2)]
Consider two linear risk sharing schemes $\boldsymbol{\xi}_1$ and $\boldsymbol{\xi}_2$ of $\boldsymbol{X}$. $\boldsymbol{\xi}_1$ weakly dominates $\boldsymbol{\xi}_2$, denoted $
\boldsymbol{\xi}_{2}\prec_{WBCX}\boldsymbol{\xi}_{1}$ if and only if there is a column-stochastic $n\times n$-matrix $C$,  such that $
\boldsymbol{\xi}_{2}=C\boldsymbol{\xi}_{1}$ and such that $\text{trace}[\text{Var}[\boldsymbol{\xi}_{2}]]\leq \text{trace}[\text{Var}[\boldsymbol{\xi}_{1}]]$.
\end{definition}

The weaker order, $\text{trace}[\text{Var}[\boldsymbol{\xi}_{2}]]\leq \text{trace}[\text{Var}[\boldsymbol{\xi}_{1}]]$, means that not necessarily  all agents prefer scheme $2$ over $1$, individually, but {\em globally}, scheme $2$ is preferred. The consequences thereof will be discussed in the subsequent section as well. A natural case would be a group that prefers the social optimum, as discussed in the introduction, a social welfare programme.

\section{Nonlinear and truncated risk sharing}\label{sec:nonlinear}

\begin{comment}
ToDo:
\begin{itemize}
    \item Reconsider the introduction, maybe too long
    \item Decide on what to move to the appendix
    \item Also make the link to \cite{achman2001insuring}
    \item Emphasise paradigm change where in traditional OR research sparsity is often the goal but in insurance that will actually change (\cite{torri2018robust} - sparse banking networks, here it is inference, \cite{ghaddar2019power} - sparse power networks)
\end{itemize}
\end{comment}

In the previous section we discussed the general case of a P2P insurance, this section extends the framework to the specific nature of the P2P markets. The conditions for P2P products often arise from a failure of the traditional providers. For example, in P2P lending, the market for loans is characterized trough mini or micro lending (c.f. \cite{iyer2016screening}) for endeavours that are not funded by traditional financial intermediaries. In the non-P2P market, some low-risk borrowers are dropping out of the market due to the demand for high collateral (see eg. \cite{liu2020social}) and have the need to be serviced by the corresponding P2P market. Likewise, in the insurance industry, uncertainty about insurees can drive insurers to demand either high premia or high deductibles which might lead to market failures. We first discuss the implications of deductibles specifically before proposing a P2P mechanism that works in such settings. 

Deductibles are also a popular technique insurance companies use to share costs with policyholders when they claim a loss, in order to reduce moral hazards and present lower premia. Insurers expect that deductibles help mitigate the behavioral risk of moral hazards, meaning that either policyholder may not act in good faith, or that they may engage in risky behavior without having to suffer the financial consequences. From an actuarial and statistical perspective, given a loss $y_i$ for a policyholder, and a deductible $s$, we can write
$$
y_i = \underbrace{\min\{y_i,s\}}_{\text{policyholder}}+\underbrace{(y_i-s)_+}_{\text{insurer}}=\begin{cases}
y_i + 0\text{, if }y_i\leq s\\
s + (y_i-s)\text{, if }y_i\geq s\\ 
\end{cases}
$$
where we use $s$ for {\em self insurance}. In the context of asymmetric information, \cite{cohen2007estimating} mention that if the policyholder can chose the level of the deductible, it might be a valid measure of the underlying risk, even if it is not possible to distinguish risk aversion and the true level of the risk (known by the policyholder, and un-observable for the insurance company, without additional assumption).

One might assume that if $y_i\leq s$, some policyholders may not report the loss, since the insurance company will not repay anything (and in a no-claim bonus systems, there are strong incentives not to declare any small claim, if they don't involve third party, as discussed in \cite{arthurRisks}). 

From an economic perspective, it is usually assumed that without deductibles, some insured could be tempted to damage their own property, or act recklessly, leading to moral hazard, as discussed intensively in \cite{eeckhoudt1991increases}, \cite{meyer1999analyzing},  \cite{Halek} and in the survey by \cite{winter2000optimal}. In the context of asymmetric information, \cite{cohen2007estimating} mention that if the policyholder can chose the level of the deductible, it might be a valid measure of the underlying risk, even if it is not possible to distinguish risk aversion and the true level of the risk (known by the policyholder, and un-observable for the insurance company, without additional assumption). Finally, \cite{dionne2001deductible} also mention fraud as the most important motivation for insurance companies to introduce deductibles. Following the seminal work papers of \cite{Duarte}, \cite{Larrimore} or \cite{XuFraud2015}, trust is a key issue in peer-to-peer risk sharing, that might yield two contradictory behaviors. On the one hand, \cite{Artis} shows that known relatives and friends might fraud together; on the other hand, \cite{Albrecher} claims that fraudulent claims should decrease, and \cite{Paperno} claims that bad-faith practices can be significantly mitigated by implementing a peer-to-peer (classical anti-fraud measures being {\em costly and hostile}).

But even if there are many theoretical justifications for introducing deductibles, they can be an important financial burden for policyholders. Since the function $s\mapsto \mathbb{E}\big[(Y-s)_+\big]$ is decreasing, the higher the deductible $s$, the smaller the premium. This again is similar to the case where a high collateral is demanded for a loan. In what follows we will study such a case and propose a P2P mechanism that acts as a complement to traditional insurance with deductibles, akin in spirit to what \cite{liu2020social} studied for P2P lending.

We consider the case where any insured can purchase an insurance contract with a fixed deductible $s$, so that the random wealth of insured $i$ at the end of the year is:
$$
X_i = Z_i\cdot\min\{s,Y_i\} = \begin{cases}
0 \text{ if no claim occurred }(Z_i=0)\\
\min\{s,Y_i\} \text{ if a claim occurred }(Z_i=1)\\
\end{cases}
$$

The mechanism that we study is then based on the following idea:  policyholders purchase insurance contracts with a deductible $s$ (that will be less expensive than having no deductible), and they consider a first layer of collaborative insurance, with peers (or friends), as in Figure \ref{fig:layers}, with some possible self retention. If this layer is a risk sharing among homogeneous peers, this will have no additional cost for the insured, but it can lower individual uncertainty. If the P2P mechanism is based on a network of friends, this can be linked to the issues (1) and (2) from the introduction. Due to the homophily property, it seems reasonable to assume that most connections possess a similar risk profile and the social collateral should enable the generation of trust between the peers and make fraud less likely.

\begin{figure}[!ht]
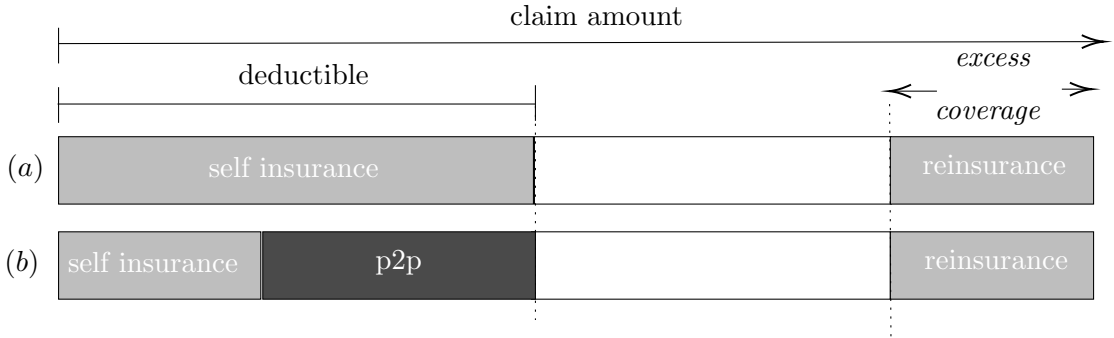

    \centering
    
    \include{tikz/info_p2p_scheme}
    \caption{
    $(a)$ is the standard 3 layer scheme of insurance coverage,
     with self insurance up to some deductible, then traditional insurance up to some upper limit, and some excess coverage provided by the reinsurance market. $(b)$ is the scheme we study in this section, with some possible first layer of self-insurance, then a peer-to-peer layer is introduced, between self insurance and traditional insurance. With our design, on a regular network, a full coverage of that second layer is possible for all participants as suggested in previous research. Above the deductible, claims are paid using traditional insurance which introduces the truncation. 
    }
    \label{fig:layers}
\end{figure}

\begin{comment}
\cite{dionne2001deductible} also mention fraud as the most important motivation for insurance companies to introduce deductibles. Following the seminal work papers of \cite{Duarte}, \cite{Larrimore} or \cite{XuFraud2015}, trust is a key issue in peer-to-peer risk sharing, that might yield two contradictory behaviors. On the one hand, \cite{Artis} shows that known relatives and friends might fraud together; on the other hand, \cite{Albrecher} claims that fraudulent claims should decrease, and \cite{Paperno} claims that bad-faith practices can be significantly mitigated by implementing a peer-to-peer (classical anti-fraud measures being {\em costly and hostile}).

\end{comment}

The issue with considering the amount below the deductible is that it essentially applies truncation to the risks, hence the convex orderings derived above might not hold. To extend our framework to this case we need to analyse such a situation numerically. A tractable way to compare risks will simply be based on the variances of risk sharing mechanism. For that we consider the following definition:

\begin{definition}[Weak Ordering of Nonlinear risk sharing schemes (1)]\label{variance:def:order}
Consider two risk sharing schemes $\boldsymbol{\xi}_1$ and $\boldsymbol{\xi}_2$ of $\boldsymbol{X}$. $\boldsymbol{\xi}_1$ weakly dominates $\boldsymbol{\xi}_2$ if $\text{trace}[\text{Var}[\boldsymbol{\xi}_{2}]]\leq \text{trace}[\text{Var}[\boldsymbol{\xi}_{1}]]$.
\end{definition}

The results from the previous sections are valid for an arbitrary adjacency matrix, our simulations will hence need to cover a wide range of possible networks to ensure comparability. We first describe our simulation process before introducing the mechanism that optimizes the non-linear case.

\begin{comment}

The convex orders derived in the previous section. 

Given that often a deductible $s$ is introduced in P2P insurance schemes to cover large losses and total contributions of a single agent tend to be capped, the risk sharing scheme becomes truncated. Hence, the convex orders derived above might not hold in every case. To extend our framework we therefore analyse such a situation numerically. A tractable way to compare risks will simply be based on the variances of risk sharing mechanism. For that we consider the following definition:

\begin{definition}[Weak Ordering of Nonlinear risk sharing schemes (4)]\label{variance:def:order}
Consider two risk sharing schemes $\boldsymbol{\xi}_1$ and $\boldsymbol{\xi}_2$ of $\boldsymbol{X}$. $\boldsymbol{\xi}_1$ weakly dominates $\boldsymbol{\xi}_2$ if $\text{trace}[\text{Var}[\boldsymbol{\xi}_{2}]]\leq \text{trace}[\text{Var}[\boldsymbol{\xi}_{1}]]$.
\end{definition}
\end{comment}

\subsection{Simulation setup}

Whereas the results from the previous section are valid for any adjacency matrix 
$A$, we need to analyse different {\em random} networks to cover a wide range of possible networks. In the graph literature, there are different approaches to simulate such networks. For example, random graph based on the classical Erd\H{o}s–Rényi (see \cite{ER59}), where each edge $\{i,j\}, i\neq j$ in $\mathcal{V}\times \mathcal{V}$ is included in the network with the same probability $p$, independently from every other edge. Then $\boldsymbol{d}$ has a binomial distribution $\mathcal{B}(n-1,p)$ which can be approximated by a Poisson distribution $\mathcal{P}(p/(n-1))$, when $n$ is large, and $p$ is not too large (note that $p$ may depend on $n$). In this case,  $\text{Var}[\boldsymbol{d}]\sim\mathbb{E}[\boldsymbol{d}]$. While this model still is interesting to study, the issue with it is that most real world examples of (social) networks have degree vectors who's distribution resembles a power law, i.e., a distribution with very thick tails. Other networks simulation methods such as the preferential attachment method \cite{Barabasi_albert_1999} can generate a degree vector $\boldsymbol{d}$ such that the distribution of the degrees follows a power law, as it creates hubs that seem common in social networks. This leads to the method having difficulties creating regular networks though. Instead of opting for a single graph generating algorithm, we propose as simulation mechanism that can handle a variety of different cases. First, draw a random degree vector according to:
{ %\color{red} @Arthur - c'est bien $d_i$ et pas $\boldsymbol{d}$
\begin{equation}\label{sim:distribution}
    d_i \overset{\mathcal{L}}{=}\min\{5+[\Delta],(n-1)\},~ \forall i,
\end{equation}
}
where $\Delta$ follows a rounded and shifted gamma distribution with mean $\mu - 5$ and variance $\sigma^2$. The shift and minimum condition ensures that every node has at least five connections and no multiple edges exist. The advantage of this approach is that by changing the parameter $\sigma$, degree distributions that resemble a Poisson distribution can be generated just as well as ones that resemble an exponential distribution (i.e.\ models generated by the Erd\H{o}s–Rényi model or the preferential attachment method, respectively). This allows to test the proposed framework in a variety of settings by changing one parameter only. In the extreme case this will result in graph building a regular mesh on the left hand side ($\sigma=0$) or a power-law degree distribution graph ($\sigma \gg 0$). 

\begin{figure}[!ht]
    \centering
         \includegraphics[width=.7\textwidth]{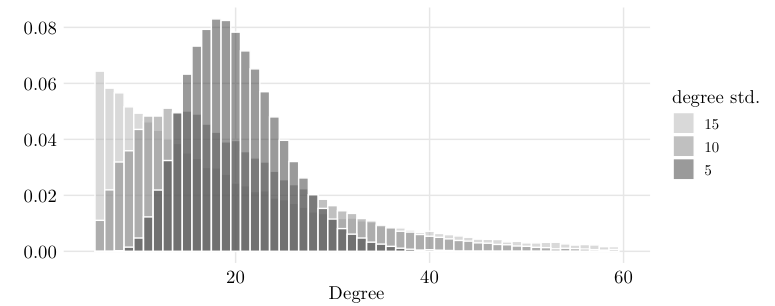} 
    \caption{Distribution of $\boldsymbol{d}$, when $\overline{d}=20$, with three different standard-deviations: 5, 10 and 15. Such distributions will be used in the numerical applications. Note that the support of $\boldsymbol{d}$ is here $\{5,6,\cdots,1000\}$, where rounded and shifted Gamma distributions were used.}
    \label{fig:dist:deg}
\end{figure}

Figure \ref{fig:dist:deg} visualises the degree distribution from three networks generated from equation \ref{sim:distribution} with a constant mean but changing standard deviation.

For the simulations we assume an average degree $\overline{{d}}=\mu - 5 = 20$ and vary the degree standard deviation $\sigma$ from $0$ to $4\overline{{d}}$ which should cover most cases of graphs observed in real life. {%\color{red} 
Once the degree vector is simulated, whereby only vectors with $\sum_i d_i$ even are considered, we construct a network using the Havel–Hakimi algorithm that ensures that no self-loops or multiple edges are present in the generated network}. Once the network is generated, we simulate the loss occurrence (here $Z$) and its corresponding severity $\tilde{Y}$, which is then capped above by a deductible which we set as $s=1000$ and denoted by $Y$. Finally, as in the preceding sections, we assume that risk between two nodes $i,j$ is only shared via the network edges, that is risk is only shared if $\{i,j\}\in \mathcal{E}$. We summarise this in the following hypothesis:

\begin{hypothesis}\label{hyp0}
Given a network $(\mathcal{V},\mathcal{E})$, policyholder $i$ will agree to share risks with all policyholders $j$ such that $\{i,j\}\in\mathcal{E}$ (called $i$'s {\em friends}). For each policyholder, the total sum collected from the friends cannot exceed the value of the actual loss.
\end{hypothesis}

\subsection{Sharing risk with friends}

We begin by studying the simplest case, where the risks are i.i.d.\ across the whole network. Traditionally, the solution that minimizes the risk would be to share the risk across all participants, in a pooled case this would amount to every agent assuming a share of $\frac{1}{n}$. In the case of sharing on networks we will show that such an approach would lead to sub-optimal outcomes once we depart from regular networks. For that, consider a regular network with $\overline{{d}}=20$ and $\text{Var}[\boldsymbol{d}]=0$, where risk is shared via reciprocal commitments as would be the case in the traditional pooled case. Sharing  the risk will all friends results in a weight $\gamma = s/\overline{d}$ which is the network equivalent of the $\frac{1}{n}$-rule. We also impose the condition that no-one should profit from a loss, hence the following definition of a reciprocal contract:

\begin{definition}\label{def:reciprocal:1}
A reciprocal contract between two policyholders, $i$ and $j$, with magnitude $\gamma$ implies that $i$ will pay an amount $C_{i\rightarrow j}$ to $j$ if $j$ claims a loss, with $C_{i\rightarrow j}\in[0,\gamma]$, and where all friends who signed a contract with $j$ should pay the same amount, and conversely from $j$ to $i$. Thus, $C_{i\rightarrow j}$ can be denoted $C_j$, and for a loss $y_j$,
$$
C_j=C_{i\rightarrow j}= \min\left\lbrace\gamma,\frac{\min\{s,y_j\}}{d_j}\right\rbrace=\min\left\lbrace\gamma,\frac{x_j}{d_j}\right\rbrace,~\forall i\in\mathcal{V}_j,
$$
where $x_i=\min\{y_i,s\}$ is the value of the loss that will be shared among friends (up to limit $\gamma$).
\end{definition}

Let $\mathcal{V}_i$ denote the set of friends from node $i$. If $i$ claims a loss, its connections will each pay a share up to $\gamma$ towards the costs and conversely, for every $j \in \mathcal{V}_i$, $i$ will pay a share in case they claim a loss. This results in the following risk sharing:
\begin{equation}\label{eq:rs1}
\xi_i = Z_i\cdot\min\{s,Y_i\} +\sum_{j\in \mathcal{V}_i} Z_j \min\left\{\gamma,\frac{\min\{s,Y_j\}}{d_j}\right\}-Z_i\cdot \min\{d_i\gamma,\min\{s,Y_i\}\}
\end{equation}
The first term is the part of the loss below the deductible $s$, because of the claim experienced by insured $i$, $X_i=Z_i\cdot\min\{s,Y_i\}$. The third term is a gain, in case $Z_i=1$ because all connections will give money, where all friends will contribute by paying $C_i$: insured $i$ cannot receive more than the loss $X_i$, that cannot be smaller than sum of all contributions $d_i\gamma$. Note that this is a risk sharing principle

\begin{proposition}
The process described by Equation (\ref{eq:rs1}) is a risk sharing principle.
\end{proposition}

\begin{proof}
With slightly simplified notation $\xi_i = X_i+S_i - Z_id_i C_i$, where $S_i$ is the total amount of money collected {\em by} policyholder $i$, while $Z_id_i C_i$ is the money given {\em to} policyholder $i$. Observe that
$$
\sum_{i=1}^n S_i = \sum_{i=1}^n \sum_{j\in \mathcal{V}_i} Z_j C_j=  \sum_{i=1}^n \sum_{j=1}^n A_{i,j}Z_j C_j=\sum_{j=1}^n \sum_{i=1}^n A_{i,j}Z_j C_j = \sum_{j=1}^n Z_j C_j  \sum_{i=1}^n A_{i,j}=\sum_{j=1}^n Z_j C_j d_j
$$
so
$$
\sum_{i=1}^n \xi_i =\sum_{i=1}^n X_i+S_i - Z_id_i C_i = \sum_{i=1}^n X_i
$$
so the process described by Equation (\ref{eq:rs1}) is a risk sharing principle.
\end{proof}

In the case of regular networks $d_i$ is the same for every agent. But if we depart from that case finding a $\gamma$ that corresponds to the $\frac{1}{n}$ approach in pooled insurance becomes more difficult. Given that the average degree according to \eqref{sim:distribution} does not change, we can consider $\gamma = s/\overline{d}$ even in the non-regular network case, which would correspond to the "average" $\frac{1}{n}$-rule.

We generate the losses by drawing i.i.d. Bernoulli variables with probability $p=1/10$, and claim size such that $Y\overset{\mathcal{L}}{=}100+\tilde{Y}$ where $\tilde{Y}$ has a Gamma distribution, with mean $\mu=900$ and variance $2000^2$, so that $Y$ has mean $1000$ and variance $2000^2$. In case of no insurance below the deductible, this will result in $\mathbb{E}[X]=45.2$ and $\text{stdev}[X]=173$ which corresponds to roughly $4.5\%$ and $17.3\%$ of the deductible respectively. We run the simulations for networks with $\sigma=\{0,1,2,\dots, 80\}$ and present the results visually in Figure \ref{fig:std1}.

\begin{figure}[!ht]
    \centering
    \includegraphics[width=.6\textwidth]{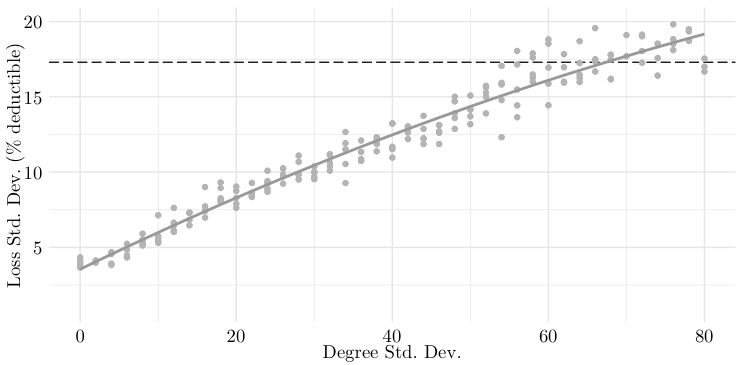}
    
    \vspace{.25cm}
    
    \caption{Evolution of $\text{stdev}[\xi]$ when the standard deviation of $D$ increases from $0$ to $80$ (and $\overline{d}=20$, using the design according to equation \eqref{sim:distribution}), using simulations of random networks of size $n=5000$. The dashed line is the standard deviation of individual losses if no risk sharing is considered. The line with the positive slope represents the calculated loss standard deviation using the mechanism described above}
    \label{fig:std1}
\end{figure}

As can be seen, the risk sharing mechanism reduces the loss standard deviation significantly for rather regular networks with low degree variance (in the extreme case where $\sigma=0$ this results in around $3.9\%\approx 17.3\%/\sqrt{20}$). But on the other extreme, where networks possess a degree vector with a distribution similar to a power law, the simple sharing mechanism actually performs {\em worse} than the no insurance case. This counterintuitive result comes from the concentration stemming from the high degree variance. There will be some nodes connected to almost every other node in the network, which end up with too many connections. Clearly, this is not a desirable risk sharing mechanism. 
%===============================================

\subsection{Optimal reciprocal engagements}\label{sec:optimalrisksharing}

So far, we assumed that all reciprocal contracts were identical. A possibility to improve on that is to assume instead that they can be unique for every edge. In that sense, the main constraint is that some connected agents $i$ and $j$, agree on the same amount of money that should be exchanged given a claim (since contracts are reciprocals), which we will denote $\gamma_{(i,j)}$. We note that for an optimization it would be almost impossible to assume that all agents will individually optimise some criteria, independently of the other. If $i$ is connected to $j$, the commitment $\gamma_{(i,j)}$ should be optimal for $i$ (thus, function of other $\gamma_{(i,\cdot)}$'s) but also optimal for $j$ (thus, function of other $\gamma_{(\cdot,j)}$'s). If there is a path between any two nodes, this would quickly become impossible to solve for a single node. A natural idea will be to consider some global planner (the P2P insurance company), maximizing some overall criterion for a given network instead. This in turn will relate to the weaker orderings that we derived above, here just in the case where the risks are truncated. We also need to extend definition \ref{def:reciprocal:1} to accommodate for this new approach. 

\begin{definition}\label{def:reciprocal:2}
A reciprocal contract between two policyholders, $i$ and $j$, with magnitude $\gamma_{(i,j)}$ implies that $i$ will pay an amount $C_{i\rightarrow j}$ to $j$ if $j$ claims a loss, with $C_{i\rightarrow j}\in[0,\gamma_{(i,j)}]$, and where $C_{k\rightarrow j}$'s (for various $k$'s) should be proportional to the engagement $\gamma_{(k,j)}$, and conversely from $j$ to $i$. Thus, for a loss $y_j$
$$
C_{i\rightarrow j} = \min\left\lbrace\gamma_{(i,j)},\frac{\gamma_{(i,j)}}{\overline{\gamma}_j}\cdot \min\{s,y_j\}\right\rbrace,\text{ where }\overline{\gamma}_j=\sum_{k\in \mathcal{V}_j}\gamma_{(k,j)}
$$
$\boldsymbol{\gamma}=\{\gamma_{(i,j)},~(i,j)\in\mathcal{E}\}$ is the collection of magnitudes.
\end{definition}

If we try to optimize this problem, it seems natural that in the high degree variance case not every node can share all of its risk. Hence a social planner could try to optimize the {\em overall} coverage within the network. This has the advantage that it is solvable by a linear program. For example, the social planner can maximize the total contributions across the networks, but with the constraint that no-single agent needs to take on more than the expected share. That is, the sum of the reciprocal engagements should not exceed the insurance need, as defined by the deductible. In the spirit of small contributions, we can also limit the amount $\gamma$ of an individual contract. A natural bound for this will be $\gamma = s/\overline{d}$, such that on average a agent can fully insure itself. Formally, we consider the following linear programming problem,
\begin{equation}\label{eq:LP}
\begin{cases}
\max\left\lbrace\displaystyle{\sum_{(i,j)\in\mathcal{E}}\gamma_{(i,j)}}\right\rbrace\\
\text{s.t. }\gamma_{(i,j)}\in[0,\gamma],~\forall (i,j)\in\mathcal{E}\\
\hspace{.65cm}\displaystyle{\sum_{j\in\mathcal{V}_i}}\gamma_{(i,j)}\leq s,~\forall i\in\mathcal{V}\\
\end{cases}
\end{equation}
where $\gamma$ is some global limit on the amount of a single contract and $s$ is the deductible of the insurance contract and assumed to be the same across all nodes, so that the second constraint is simply related to Hypothesis \ref{hyp0}.
With classical linear programming notations, we want to find $\boldsymbol{z}^\star\in\mathbb{R}_+^m$ where $m$ is the number of edges in the network (we consider edges $(i,j)$, with $i<j$), so that
$$
\begin{cases}
\boldsymbol{z}^\star =\displaystyle{\underset{\boldsymbol{z}\in\mathbb{R}_+^m}{\text{argmax}}\left\lbrace\boldsymbol{1}^\top\boldsymbol{z}\right\rbrace}\\
\text{s.t. }\boldsymbol{A}^\top\boldsymbol{z}\leq \boldsymbol{a}
\end{cases} 
$$
when $\boldsymbol{A}$ is a $(n+m)\times m$ matrix, and $\boldsymbol{a}$ is a $(n+m)$ vector,
$$
\boldsymbol{A}=\left[\begin{matrix}
\boldsymbol{T} \\ \mathbb{I}_m
\end{matrix}\right]\text{ and }\boldsymbol{a}=\left[\begin{matrix}
\gamma\boldsymbol{1}_n \\ s\boldsymbol{1}_m
\end{matrix}\right]
$$
where $\boldsymbol{T}$ is the incidence matrix, $\boldsymbol{T}=[T_{u,w}]$, with $w\in\mathcal{V}$ i.e. $w=(i,j)$, the $T_{u,w}=0$ unless either $w=i$ or $w=j$ (satisfying $\boldsymbol{d}=\boldsymbol{T}\boldsymbol{1}_n$).

Further, as discussed in the supplementary material, it is possible to retain some level of self contribution for every node instead of a sharing the entire risk, this is mainly introduced to achieve ex-post fairness. This self contribution needs to be paid first before any other peers will contribute to a loss, to ensure fairness. Any amount above the self contribution but below the deductible will then be shared pro-rata between the connected peers. Figure \ref{fig:schemepyramid} depicts the insurance scheme in detail.
\begin{figure}[!ht]
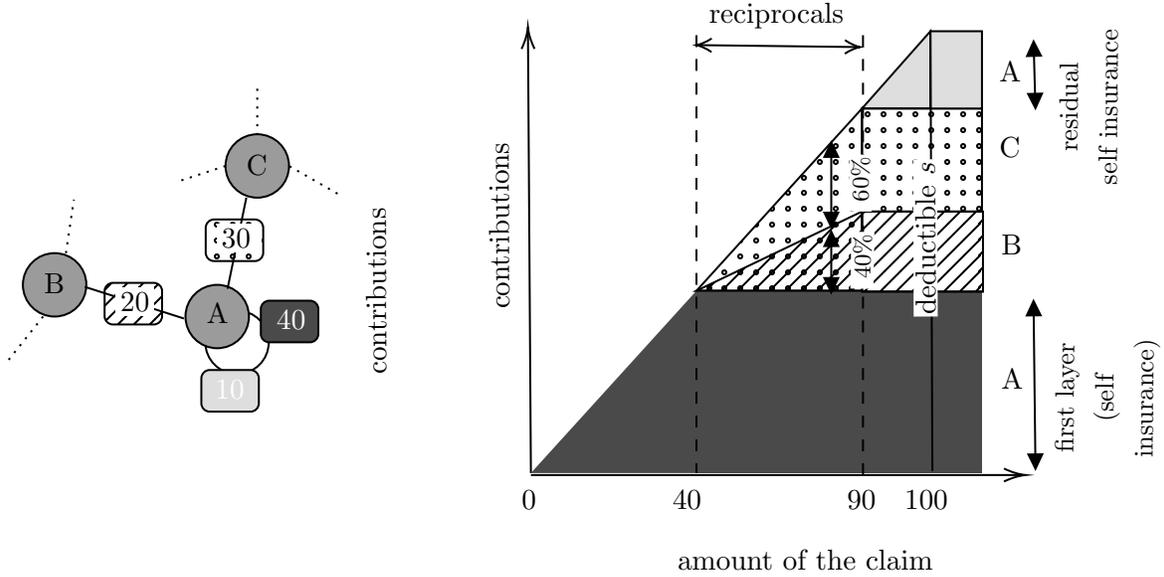

    \centering
    
    \include{tikz/figure_pyramide}
    \caption{
    Insurance scheme for node $A$ in the network left. First, a self contribution will need to be payed. Then the amount between the self insurance and the deductible is shared pro rata between the connected nodes. Finally, if node $A$ cannot share all its risk, there will be some residual self-contribution. This will only need to be payed if the claim exceeds the sum of the self insurance and reciprocal contracts.
    }
    \label{fig:schemepyramid}
\end{figure}

The results of the simulation are summarised graphically in figure \ref{fig:std3}. With optimized reciprocal engagements the curves do not intersect (or even exceed) the no-insurance case. This should be expected as sharing among i.i.d. risks can only decrease the expected variance. It is worth to study the economic framework behind this though.

\begin{figure}[!ht]
    \centering
    \includegraphics[width=.6\textwidth]{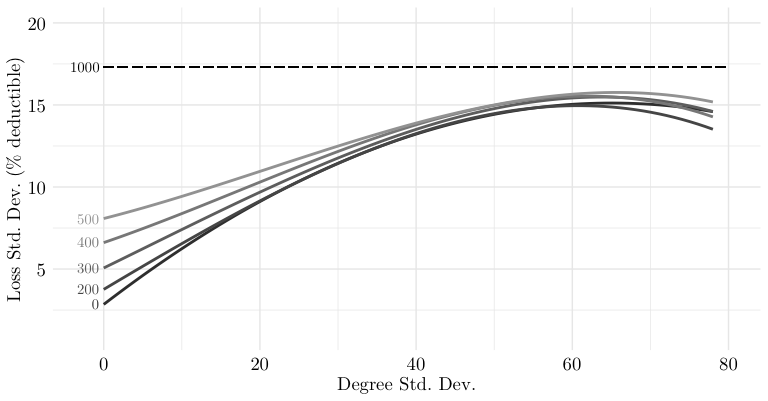}
    \caption{Evolution of $\text{stdev}[\xi^\star]$ when the standard deviation of $D$ increases from $0$ to $80$ (and $\overline{d}=20$), using simulations of random networks of size $n=5000$. Here the deductible is $s=1000$, and various scenarios are considered.} %As in Figure \ref{fig:std2}     \textcolor{red}{Reference missing?}five levels of self contributions are considered, at $0$ (no self-contribution) $200$, $300$, $400$ and $500$.}
    \label{fig:std3}
\end{figure}

Even with the simple linear program, we can consider the social planner to be optimizing the sum of utilities, which results in a Benthamite (or utilitarian) welfare function. Assuming each utility is on the same scale, each agent has a utility linear in the sum of the reciprocal contracts. That is, the preferences are strictly monotone until $s$ and only weakly convex. The issue with optimizing the welfare function this way is that the solution is not unique, nor does it necessarily minimize the variance across the portfolio, even in the i.i.d.\ case. Instead of optimizing only the coverage, we can instead consider optimizing the welfare function using a different utility framework. As the next section shows, this will result in more stable solutions and also allows for heterogeneous and potentially correlated risks but comes at the cost of higher computational complexity.

\subsection{Utility maximizing framework}

Depending on the utility function chosen, the social planners objective needs to be adapted. A simple adaption is to assume a quadratic utility for each agent. In this case, the optimization problem can be solved with a quadratic program, under the same constraints as before. Note here that self-insurance arises naturally if we consider the edges $\{i,i\}$ to be existent for every $i=1,\dots,n$. This has the natural interpretation for concave utilities that agents will prefer to mix up the risk stemming from themselves with that of others to maximize their utility.

\begin{figure}[!ht]
    \centering
    \includegraphics[width=.6\textwidth]{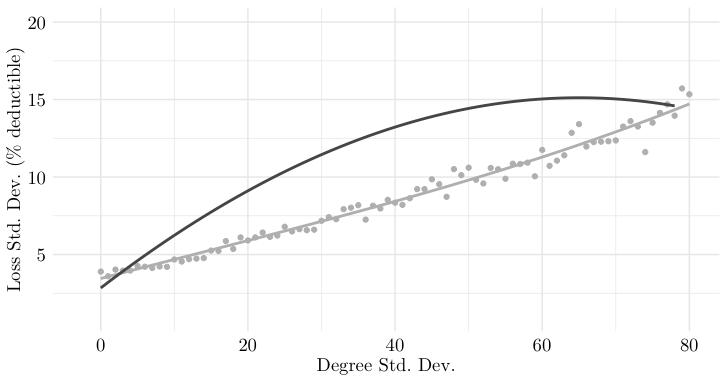}
    \caption{Evolution of $\text{stdev}[\xi^\star]$ when the standard deviation of $D$ increases from $0$ to $80$ (and $\overline{d}=20$) using the quadratic utility maximizing approach in light gray and the results from \ref{fig:std3} as comparison in black} %As in Figure \ref{fig:std2}     \textcolor{red}{Reference missing?}five levels of self contributions are considered, at $0$ (no self-contribution) $200$, $300$, $400$ and $500$.}
    \label{fig:quadratic utility}
\end{figure}

Using the quadratic utility function, we can translate this into a quadratic optimization as follows:

\begin{align*}
    &\min_{w} w^\top Q w \\
    \text{st.}\quad &\sum_{j \notin \mathcal{V}_i} w_{i,j} = 0, \quad \forall i \\
    &\sum_{j \in \mathcal{V}_i} w_{i,j} = 1, \quad \forall i \\
    & w_k \geq 0\quad \forall k \\
    & w_{i,j}s \leq \gamma \quad \forall j \neq i\\
    & w_{i,i} \leq 1
\end{align*}

Where $w$ is the $(\frac{n(n+1)}{2} \times 1)$ vector of (individual) edges, Q is a diagonal matrix of size $(\frac{n(n+1)}{2} \times \frac{n(n+1)}{2})$ where the elements on the diagonal are $1$ if $i=j$ and $2$ otherwise. The solution to this problem guarantees a minimum variance globally. Further, the solution to the optimization also guarantees that the proposed mechanism will work at least as well as the no-insurance case (as the case where $w_{i,i}=1$ is within the support).

Note that using the utility framework would also allow to incorporate heterogeneous risks and possible correlations between the risks. The issue is that they have to be known to the social planner, which is similar to an insurance company having {\em all} relevant risk-factors available. We refrain from this hypothesis in the spirit of P2P products, where often exactly this information is not available (after all - why would the insurance impose a deductible if the risks would be perfectly known). Instead, we focus on the sharing mechanisms themselves and assume relative homophily based on the network properties alluded above. 

\subsection{Extension to lower trust}

The mechanism effectively lowers the individual uncertainty without an additional cost to the agents, but it might still be improved though. Especially in networks that are highly concentrated many nodes are insured to a suboptimal degree. Highly concentrated graphs have the issue that a few nodes are connected to many nodes but {\em most} of the agents are not connected to enough others to share their risk with as many as possible. On the other hand, networks with a higher degree variance can be considered more dense in another way. For example star shaped networks where the degree variance is maximal are used to ensure quick routing between different computer nodes, as all only need to pass the central node to arrive at any other node (note that this central node is conceptually similar to an insurer in the non-P2P case).

Further, just as \cite{lin2013judging} made an explicit link to hierarchy and social collateral with friends in the P2P lending literature, we propose a framework similar to that. Instead of merely sharing risk with friends, it can be possible to share some risk with the friends of a friend (corresponding to all connections with path length 2). Just as in the case for loans, trust (as expressed by social collateral) will be lower between friends-of-friends than with friends. A natural way to deal with this is to limit the maximum contractual values between friends-of-friends to a lower amount. With some self insurance, a new mechanism is depicted in Figure \ref{fig_insured_BCE_self_ff}, where the dashed line is a connection between friends-of-friends. To incorporate this into the linear program above, we state the following hypothesis:
\begin{hypothesis}\label{hyp2}
Consider a policyholder $i$ such that his coverage $\gamma^\star_i$ is strictly smaller than the overall loss limit $s$ then policyholder $i$ will try get connections with friends of friends.  
\end{hypothesis}

The optimization problem then needs to have a definition for friends of friends:

\begin{definition}\label{def:friendsoffriends}
 Given a network $(\mathcal{V}, \mathcal{E})$, with degrees $\boldsymbol{d}$ and some first-level optimization vector $\boldsymbol{\gamma}_1$, define the subset of vertices who could not share their whole deductible as, 
 $$
 \mathcal{V}_{\boldsymbol{\gamma}_1} = \big\lbrace i \in \mathcal{V}: \gamma_i=\sum_{j\in\mathcal{V}}\gamma_{(i,j)} < s \big\rbrace.
 $$
 The $\boldsymbol{\gamma}_1$-sub-network of friends of friends is $(\mathcal{E}_{\boldsymbol{\gamma}_1}^{(2)},\mathcal{V}_{\boldsymbol{\gamma}_1})$,
$$
\mathcal{E}_{\boldsymbol{\gamma}_1}^{(2)} = \big\lbrace(i,j)\in\mathcal{V}_{\boldsymbol{\gamma}_1}\times\mathcal{V}_{\boldsymbol{\gamma}_1}:\exists k\in\mathcal{V}\text{ such that }(i,k)\in\mathcal{E},~(j,k)\in\mathcal{E}\text{ and }i\neq j\big\rbrace
$$
 
\end{definition}

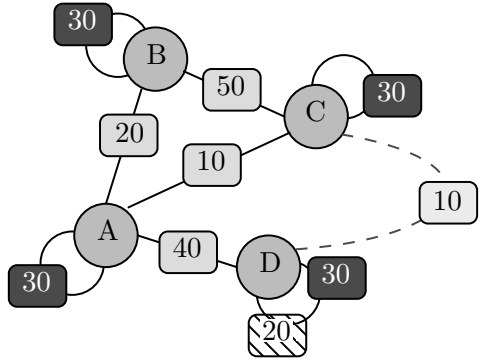
\begin{figure}[!ht]
    \centering
  \input{tikz/dessin8}
  \vspace{.5cm}
    \caption{Network with five nodes $\{$A,B,C,D,E$\}$, with some self-contribution, and some possible {\em friends-of-friends} risk sharing (with a lower maximal contribution).}
    \label{fig_insured_BCE_self_ff}
\end{figure}

Again, from Hypothesis \ref{hyp2}, policyholders with enough connections (to recover fully from a loss) will not need to share their risk with additional people. So only people in that sub-network will try to find additional resources through friends of friends in order to fully cover the deductible. Figure \ref{fig:friends5005} depicts the total share of the risk share with either friends, friends-of-friends or retained as residual self insurance. The part shareable with friends decreases with increasing degree variance, whereas all of the risk can be shared with friends in the regular graph case, most of the risk will need to be assumed by an agent itself in the case of a high degree variance graph.

\begin{figure}[!ht]
    \centering
    \includegraphics[width=.6\textwidth]{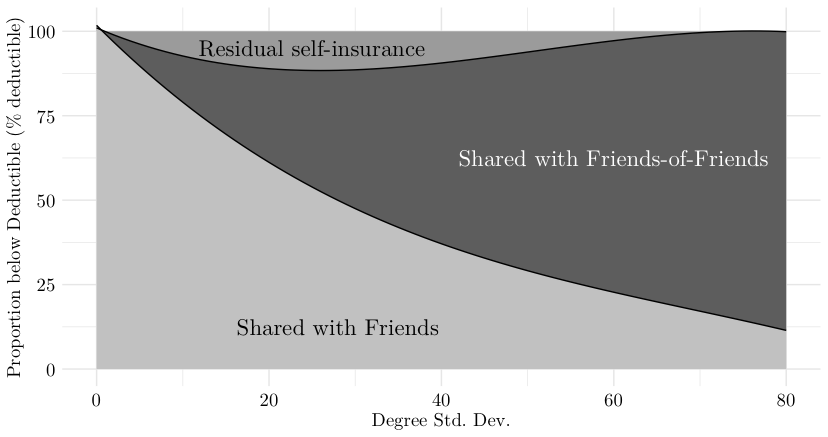}
    \caption{Distribution of the risk for $\gamma_1 = 50$ and $\gamma_2 = 5$, without any initial self-contribution. Note how the total shareable risk first decreases but then again increases again as the share of friends first drops faster than the share of friends-of-friends can increase. As networks with a high degree variance generate ``hubs" (similar to a star shaped network), these hub-nodes can help to connect most of the nodes in the network by passing through these central nodes.}
    \label{fig:friends5005}
\end{figure}

Although the graph in Figure \ref{fig:friends5005} has a second level contribution limit that is only $10\%$ of the first level limit, the relative importance of the second layer becomes larger with increasing degree variance as the more concentrated, high degree variance networks generate "hubs" to facilitate sharing between friends-of-friends. Although computationally more intensive, we can still solve the friends of friend risk sharing mechanism with a linear program, depending on the utility function that is assumed.

\subsection{Optimal friends of friends}

We propose to solve the optimization problem with a two stage approach. The first stage is the same as previously in section \ref{sec:optimalrisksharing}, where the initial contribution among friends, $\gamma_{(i,j)}$ is bounded by $\gamma_1$
$$
\begin{cases}
\gamma_1^\star=\text{argmax}\left\lbrace\displaystyle{\sum_{(i,j)\in\mathcal{E}^{(1)}}\gamma_{(i,j)}}\right\rbrace\\
\text{s.t. }\gamma_{(i,j)}\in[0,\gamma_1],~\forall (i,j)\in\mathcal{E}^{(1)}\\
\hspace{.65cm}\displaystyle{\sum_{j\in\mathcal{V}_i^{(1)}}}\gamma_{(i,j)}\leq s,~\forall i\in\mathcal{V}\\
\end{cases}
$$
In the second stage, $\gamma_{(i,j)}$ is bounded by $\gamma_2$. $\mathcal{E}_{\boldsymbol{\gamma}_1^\star}^{(2)}$ can be obtained by squaring the adjacency matrix and setting nonzero entries to 1 as defined in definition \ref{def:friendsoffriends}. 
$$
\begin{cases}
\gamma_2^\star=\text{argmax}\left\lbrace\displaystyle{\sum_{(i,j)\in\mathcal{E}^{(2)}}\gamma_{(i,j)}}\right\rbrace\\
\text{s.t. }\gamma_{(i,j)}\in[0,\gamma_2],~\forall (i,j)\in\mathcal{E}_{\boldsymbol{\gamma}_1^\star}^{(2)}\\
\hspace{.65cm}\displaystyle{\sum_{j\in\mathcal{V}_i^{(1)}}\gamma^\star_{1:(i,j)}+\sum_{j\in\mathcal{V}^{(2)}_i}}\gamma_{(i,j)}\leq s,~\forall i\in\mathcal{V}\\
\end{cases}
$$

In that way, most of the risk can be shared between peers and the results are again graphically summarised in Figure \ref{fig:std6}. As can be seen, the loss standard deviation is now well below the no-insurance case even for high degree variance networks. Again, this framework could be extended to different utility functions.

\begin{figure}[!ht]
    \centering
    \includegraphics[width=.6\textwidth]{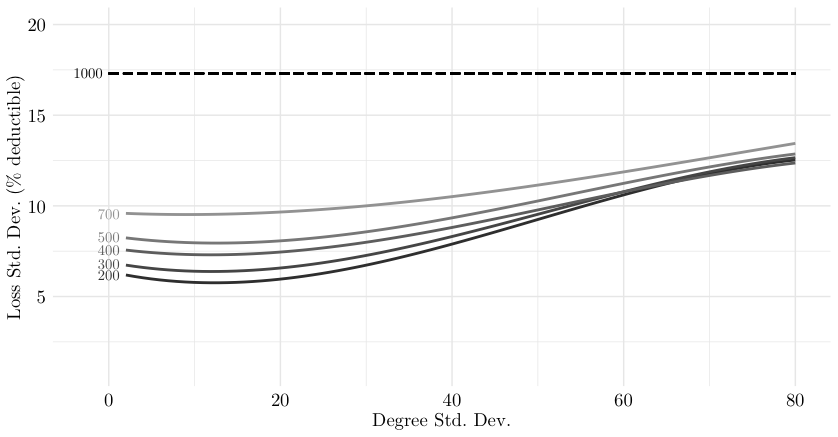}
    \caption{Evolution of $\text{stdev}[\xi]$ when the variance of $\boldsymbol{d}$ increases from $0$ to $80$ (and $\overline{d}=20$), using simulations of random networks of size $n=5000$. Eight levels of self contributions are considered, at $0$ (no self-contribution) $200$,  $300$, $400$, 500 and $700$, when friends, as well as (possibly) friends of friends, can share risks. A self contribution of $1000$ (dashed line) means that no risk is shared.}
    \label{fig:std6}
\end{figure}

\section{Conclusion}\label{sec:conclusion}

Recently, \cite{denuit2020}, \cite{Robert}, \cite{Robert2} and \cite{feng_2020_flood} studied actuarial properties of peer-to-peer insurance mechanisms on networks but restricted their analysis to specific shapes. In this paper, we consider the use of reciprocal contracts, as a technique to share risks with ``{\em friends}''. The notion of friends and social ties required us to study networks of arbitrary shapes that are described by their adjacency matrix. The use of this inspired us to derive mechanisms to compare different linear risk sharing mechanism. We then extended the framework to a situation that poses P2P insurance as a complement rather than a substitute to traditional insurance, similar to what was found to be the case in P2P lending (eg. \cite{tang2019peer}). For that we considered the situation where not the entire risk must be shared but rather the risk associated to a deductible. By assuming the availability of similar contracts with the same deductible for the policyholders, we investigated the shape of the network, or more specifically, the total number of connections (or equivalently the average degree of a node), and the impact of the variability of the degrees. We observed that the more regular the network, the more efficient the risk sharing, in the sense that it will decrease the risks for individuals. This in turn also serves as a guide when traditional insurance is most apt as opposed to the "ancient" form of mutual insurance. We proposed a crude initial mechanism to solve the P2P insurance problem under the hypothesis of the deductible using identical contracts among policyholders. The simplicity of this mechanism allows us to evaluate the dynamics with ease: policyholder might decide to leave, or join by signing new contracts with friends. This behaviour is interesting, but it is clearly sub-optimal, especially when the network is not regular. In that case, we found that there are policyholder with too many friends, who might select those they wish to sign reciprocal contracts with, and policyholders with too few friends.

To overcome that problem, we considered an optimal computation of commitments, contract per contract, from a global planer's perspective (in the sense that we simply want to maximize the overall coverage through those contracts or minimize the risk for agents with a quadratic utility). With this approach, the dynamics of the approach becomes more complex, since it is necessary to compute the optimal values when someone joins or leaves, which yields practical issues. But those contracts clearly lower the risks for policyholder, who can now share risks with their friends. Finally, we considered a possible extension, were additional edges could be considered: policyholders can also agree to share risks with friends of friends. Even if we assume that policyholders might have less interest to share risks with people they don't know, and therefore assume that the commitment will be financially much smaller, we see that it is possible to increase the total coverage, and decrease more the variability of the losses for policyholders. In the study of P2P insurance mechanisms, much remains to be done but the results and techniques derived in this article should serve as a starting point for further exploration, for example on the dynamics of such systems.

%\clearpage

% \textcolor{red}{REFERENCES A RAJOUTER : , \cite{guo2016instance}, \cite{hong2019optimal}, \cite{wang2022credit}, \cite{grechuk2015center}}

\bibliographystyle{apalike}
\bibliography{bibliography}

\section{Appendix}

\subsection{Heuristic interpretation of weaker orderings}

In this appendix, we briefly discuss what $\text{trace}[C\Sigma C^\top]>\text{trace}[\Sigma]$ could mean, in dimension 2, when $C$ is column-stochastic matrix.

\subsection{$\text{trace}[D\Sigma D^\top]$ against $\text{trace}[\Sigma]$ in dimension 2}\label{app:1}

(i) In dimension 2, we can prove that the result holds. There is $x\in[0,1]$ such that
$$
D=\begin{pmatrix}
x& 1-x\\
1-x & x
\end{pmatrix}\text{ and }\operatorname{Var}[\boldsymbol{\xi}_1]=
\begin{pmatrix}
a^2& rab\\
rab & b^2
\end{pmatrix}$$

\begin{align*}
\operatorname{trace}\big[\operatorname{Var}[\boldsymbol{\xi}_2]\big]&=\operatorname{trace}\big[D\operatorname{Var}[\boldsymbol{\xi}_1]D^\top\big] \\ 
& = \text{trace}\left[\begin{pmatrix}
x& 1-x\\
1-x & x
\end{pmatrix}\begin{pmatrix}
a^2& rab\\
rab & b^2
\end{pmatrix}\begin{pmatrix}
x& 1-x\\
1-x & x
\end{pmatrix}^\top\right] \\ 
&= \text{trace}\left[\begin{pmatrix}
x^2a^2+ 2rab(1-x)x +b^2(1-x)^2 ~~~ *\\
* ~~~ (1-x)^2a^2+ 2rab(1-x)x +b^2x^2 
\end{pmatrix}\right] \\
&\leq (a^2+b^2)\big(x^2+(1-x)^2\big)+4abx(1-x)
\end{align*}

This parabolic function (in $x$) is symmetric in $x$ and $(1-x)$, so it is symmetric in $x=1/2$, which is the minimum of that function. The maximum of that parabolic function on the interval $[0,1]$ is obtained either when $x=0$ or $x=1$, and it takes value $a^2+b^2$, so, the trace is lower than
$$
\leq a^2+b^2=\text{trace}\left[\begin{pmatrix}
a^2& rab\\
rab & b^2
\end{pmatrix}\right]
$$
% (ii) If $\boldsymbol{\xi}_1$ consists in independent risks with variance $\sigma^2$, we can prove that the result holds.
% $$
% \text{trace}\big[\text{Var}[\boldsymbol{\xi}_2]\big]=\sigma^2 \text{trace}\big[DD^\top\big]\leq \sigma^2n
% $$
% since $\text{trace}\big[DD^\top\big]$ has positive eigenvalues, and the largest one is 1. So the trace is bounded by $n$, and therefore $\text{Var}[\xi'_2]\leq\text{Var}[\xi_1']=\sigma^2$.

% (iii) If $\boldsymbol{\xi}_1$ consists in independent risks, matrix $\operatorname{Var}[\boldsymbol{\xi}_1$ is diagonal,
% $$
% \operatorname{trace}\big[D\operatorname{Var}[\boldsymbol{\xi}_1]D^\top\big]=\sum_{i=1}^n D_i^T\top \operatorname{Var}[\boldsymbol{\xi}_1] D_i=\sum_{i=1}^n \sigma_i^2 D_{i,i}^2\leq \sum_{i=1}^n \sigma_i^2=\operatorname{trace}\big[\operatorname{Var}[\boldsymbol{\xi}_1]\big]
% $$
% where the diagonal of $\operatorname{Var}[\boldsymbol{\xi}_1]$ is $\boldsymbol{\sigma^2}$.

\subsection{$\text{trace}[C\Sigma C^\top]$ against $\text{trace}[\Sigma]$}\label{app:2}

In order to understand, consider the case in dimension 2. Consider $x,y\in[0,1]$ so that
$$
\boldsymbol{\xi}_2=C\boldsymbol{\xi}_1\text{with }C=\begin{pmatrix}
x& 1-y\\
1-x & y
\end{pmatrix},\text{ while }
\text{Var}[\boldsymbol{\xi}_1]=\begin{pmatrix}
a^2& rab\\
rab & b^2
\end{pmatrix}
$$
Then, since $\text{Var}[\xi_2']=\displaystyle{\frac{1}{2}}\text{trace}[C\text{Var}[\boldsymbol{\xi}_1]C^\top]$, write $\text{trace}[C\text{Var}[\boldsymbol{\xi}_1]C^\top]$ as

\begin{align*}
    &\text{trace}\left[\begin{pmatrix}
x& 1-y\\
1-x & y
\end{pmatrix}\begin{pmatrix}
a^2& rab\\
rab & b^2
\end{pmatrix}\begin{pmatrix}
x& 1-x\\
1-y & y
\end{pmatrix}\right] \\
&= \text{trace}\left[\begin{pmatrix}
x^2a^2+ 2rab(1-y)x +b^2(1-y)^2 ~~~ *\\
* ~~~ (1-x)^2a^2+ 2rab(1-x)y +b^2y^2 
\end{pmatrix}\right] \\
&=a^2\big(x^2+(1-x)^2\big)+b^2\big(y^2+(1-y)^2\big)\\
&\quad +2rab\big[y(1-x)+x(1-y)\big]
\end{align*}

that we can write as
$$
=\frac{a^2}{2}\big[(2x-1)^2+1]+\frac{b^2}{2}\big[(2y-1)^2+1]+rab
\big[(2x-1)(2y-1)-1\big] 
$$
which is a quadratic form in $\boldsymbol{Z}=(X,Y)=(2x-1,2y-1)$
$$
=\frac{a^2}{2}\big[X^2+1]+\frac{b^2}{2}\big[Y^2+1]+rab
\big[XY-1\big] = \frac{1}{2}\boldsymbol{Z}^\top\begin{pmatrix}
a^2 & rab \\
rab & b^2
\end{pmatrix}\boldsymbol{Z}
$$

which is an elliptic paraboloid function, minimal in $(X,Y)=(0,0)$ -- or $(x,y)=(1/2,1/2)$. Since $(x,y)\in[0,1]^2$, the maximum is either attained in $(0,0)$ or $(1,1)$ when $r<0$ or $(0,1)$ or $(1,0)$ when $r>0$. In the first case, we obtain that the trace is lower than $a^2+b^2$, and in the second case, it can exceed $a^2+b^2$. More specifically, on Figure \ref{parab}, we can see that values of $(x,y)$ such that $\text{trace}[C\text{Var}[\boldsymbol{\xi}]C^\top]> \text{trace}[\text{Var}[\boldsymbol{\xi}]]$ (in red, on the right), that is either when one is close to 1, and the other close to 0.

\begin{figure}
    \centering
    \includegraphics[width=.32\textwidth]{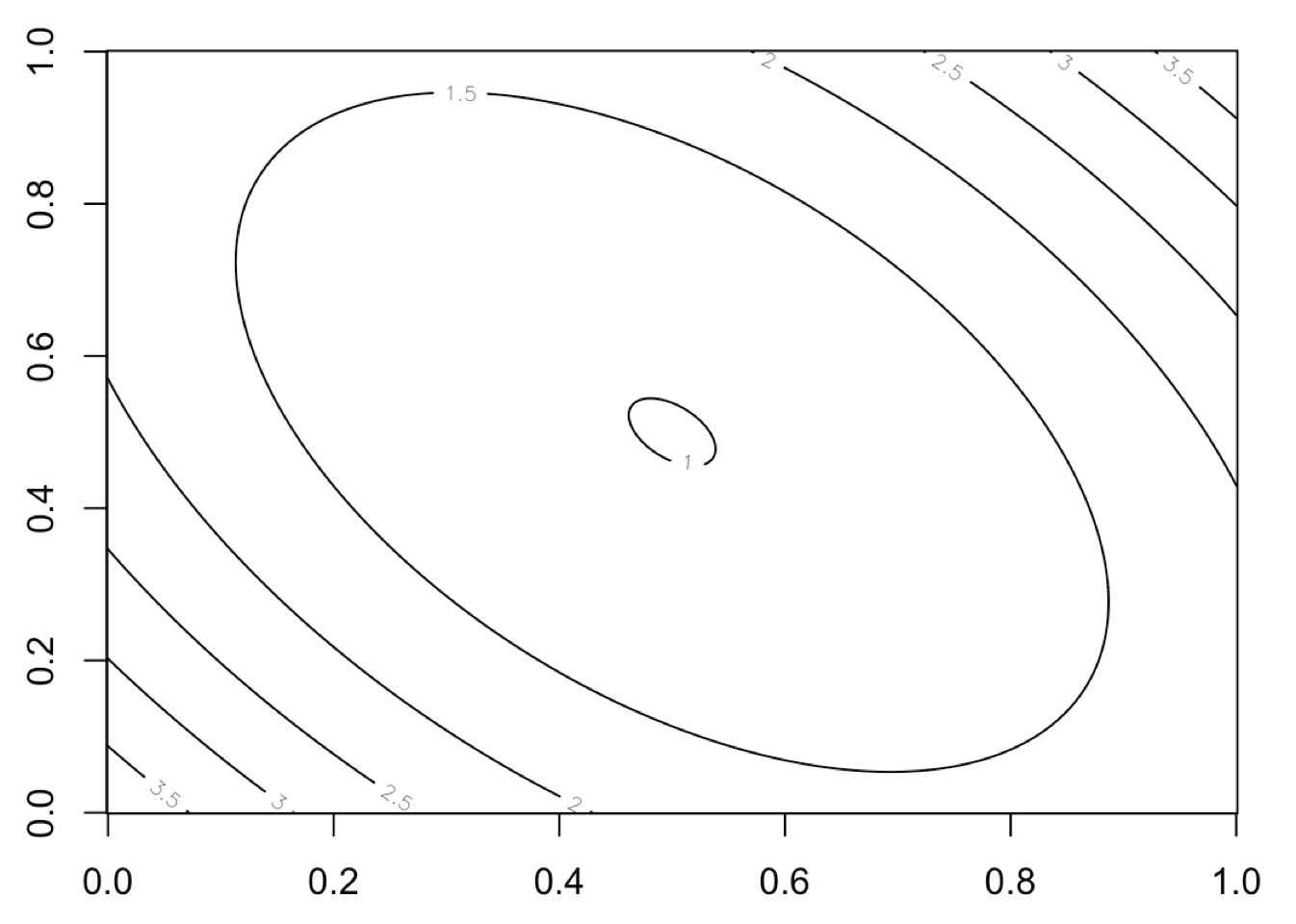}~
    \includegraphics[width=.32\textwidth]{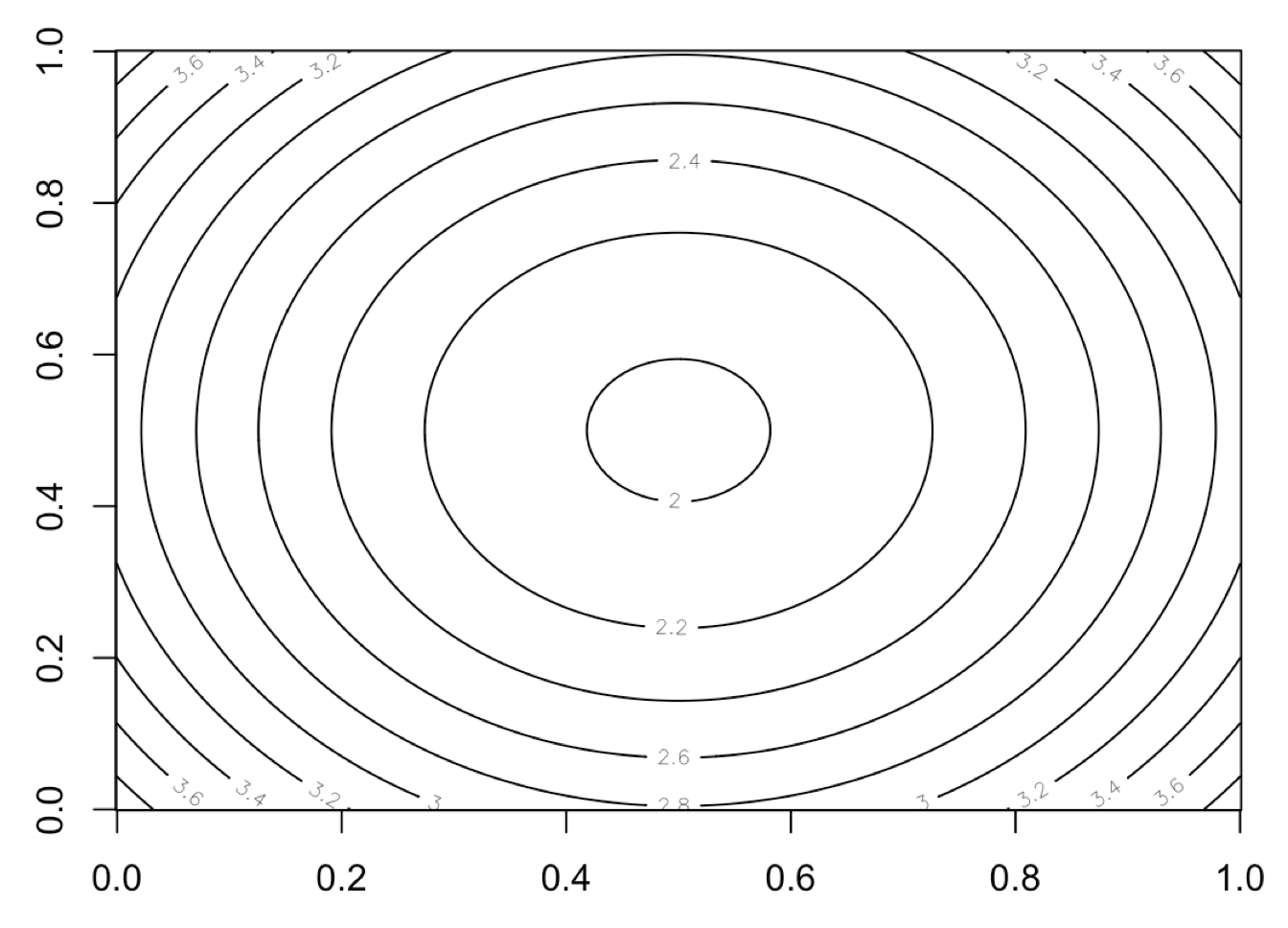}~
    \includegraphics[width=.32\textwidth]{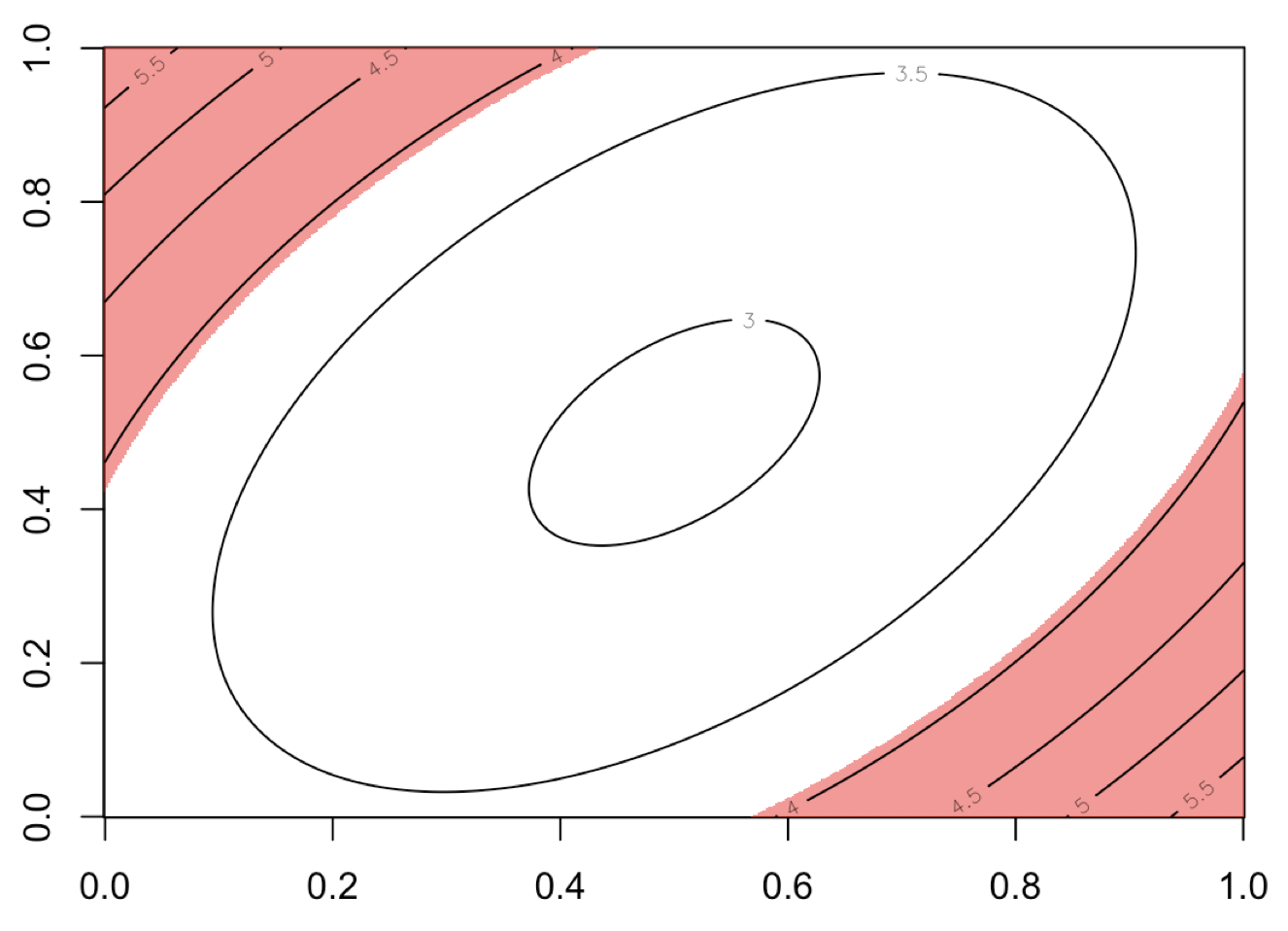}
    \caption{Level curves of $(x,y)\mapsto \text{trace}[C\text{Var}[\boldsymbol{\xi}]C^\top]$ where $C$ is the column-stochastic matrix with diagonal $(x,y)$, for some variance matrix with negative correlation on the left, no correlation in the middle, and positive correlation on the right. For a negative correlation, $\text{trace}[C\text{Var}[\boldsymbol{\xi}]C^\top]\leq \text{trace}[\text{Var}[\boldsymbol{\xi}]]$. For a positive correlation, the red are corresponds to cases where $\text{trace}[C\text{Var}[\boldsymbol{\xi}]C^\top]> \text{trace}[\text{Var}[\boldsymbol{\xi}]]$.}
    \label{parab}
\end{figure}

In that case, for instance with $x=\epsilon$ and $y=1-\epsilon$, it means that

\begin{align*}
    \begin{pmatrix}
\xi_{2,1} \\
\xi_{2,2}
\end{pmatrix} &=
\begin{pmatrix}
x & 1-y \\
1-x & y\\
\end{pmatrix}
\begin{pmatrix}
\xi_{1,1} \\
\xi_{1,2}
\end{pmatrix} \\
&=
\begin{pmatrix}
\epsilon & \epsilon \\
1-\epsilon & 1-\epsilon\\
\end{pmatrix}
\begin{pmatrix}
\xi_{1,1} \\
\xi_{1,2}
\end{pmatrix}\\ 
&=\begin{pmatrix}
\epsilon(\xi_{1,1}+\xi_{1,2})\\
(1-\epsilon)(\xi_{1,1}+\xi_{1,2})\\
\end{pmatrix}
\end{align*}

Even if it is a risk-sharing, it can be seen as an unfair, or unbalanced, one, since policyholder 2 is now taking all most of the risks.

\subsection{Risk sharing on cliques}

Such a design was briefly mentioned in \cite{feng_2020_flood}, and called {\em Hierarchical P2P Risk Sharing} although their model results by partitioning an already complete graph into several complete subgraphs. An important concept here are cliques

\begin{definition}[Clique]
A clique $\mathcal{C}_i$ within an undirected graph $\mathcal{G} = (\mathcal{V}, \mathcal{E})$ is a subset of vertices $\mathcal{C}_i \subseteq \mathcal{V}$ such that every two distinct vertices are adjacent. That is, the induced subgraph of $\mathcal{G}$ by $\mathcal{C}_i$ is complete. A \emph{clique cover} is a partition of the graph $\mathcal{G}$ into a set of cliques. 
\end{definition}\label{def:clique}

In a network, a clique is a subset of nodes such that every two distinct vertices in the clique are adjacent. Assume that a network with $n$ nodes has two (distinct) cliques, with respectively $k$ and $m$ nodes. One can consider some ex-post contributions, to cover for losses claimed by connected policyholders. If risks are homogeneous, contributions are equal, within a clique. Assume that policyholders face random losses $\boldsymbol{X}=(X_1,\cdots,X_k,X_{k+1},\cdots,X_{k+m})$ and consider the following risk sharing
$$
\xi_i = \begin{cases}
\displaystyle{\frac{1}{k}\sum_{j=1}^k X_j}\text{ if }i\in\{1,2,\cdots,k\}\vspace{.2cm}\\
\displaystyle{\frac{1}{m}\sum_{j=k+1}^{k+m} X_j}\text{ if }i\in\{k+1,k+2,\cdots,k+m\}
\end{cases}
$$
This is a linear risk sharing, with sharing matrix $D_{k,m}$, so that $\boldsymbol{\xi}=D_{k,m}\boldsymbol{X}$, defined as
$$
D_{k,m}= \qquad \bordermatrix{~  &  ~1 & ~\cdots & ~k & k+1 & k+2
                        & ~\cdots & k+m  \cr
                     ~~~1&k^{-1} & \cdots & k^{-1} & 0 & 0 & \cdots & 0  \cr
                    ~~~\vdots &  \vdots & & \vdots & \vdots & \vdots &  & \vdots \cr
                  ~~~k&k^{-1} & \cdots & k^{-1} & 0 & 0 & \cdots & 0  \cr
                    k+1 & 0 & \cdots & 0 & m^{-1} & m^{-1} & \cdots& m^{-1} \cr
                    k+2 & 0 & \cdots & 0 & m^{-1} & m^{-1} & \cdots& m^{-1} \cr
                    ~~~\vdots &  \vdots & & \vdots & \vdots & \vdots &  & \vdots \cr k+m & 0 & \cdots & 0 & m^{-1} & m^{-1} & \cdots& m^{-1} \cr
                    }
$$
Since $D_{k,m}$ is a doubly-stochastic matrix, $\boldsymbol{\xi}\prec_{CCX}\boldsymbol{X}$. In order to illustrate Proposition \ref{prop:3:7}, let $I$ denote a uniform variable over $\{1,2,\cdots,n\}$, and $\xi'=\xi_I$,
$$
\mathbb{E}[\xi']=\mathbb{E}[\mathbb{E}[\xi_I|I]]=\frac{1}{n}\sum_{i=1}^n \mathbb{E}[\xi_i]=\frac{1}{n}\sum_{i=1}^n\mathbb{E}[X_i]=\mathbb{E}[X]
$$
as expected since it is a risk sharing, while
\begin{multline*}
\text{Var}[\xi']=\mathbb{E}[\text{Var}[\xi_I|I]]=\frac{1}{n}\sum_{i=1}^n \text{Var}[\xi_i]=\frac{1}{n}\left(k\frac{\text{Var}[X]}{k^2}+m\frac{\text{Var}[X]}{m^2}\right)\\=\frac{\text{Var}[X]}{n}\left(\frac{1}{k}+\frac{1}{m}\right)=\frac{\text{Var}[X]}{k(n-k)}
\end{multline*}
since $n=k+m$, so that $\text{Var}[\xi']<\text{Var}[X]$, meaning that if we randomly pick a policyholder, the variance of the loss while risk sharing is lower than the variance of the loss without risk sharing. Observe further that $k\mapsto k(n-k)$ is maximal when $k=\lfloor n/2 \rfloor $, which means that risk sharing benefit is maximal (socially maximal, for a randomly chosen representative policyholder) when the two cliques have the same size.

{

\subsection{Cliques in practical problems}

Cliques have attractive properties, in that they allow to use a linear risk sharing mechanism on the resulting subgraphs. Further, they can be used to represent a (more) homogeneous subgroup within a larger network, as for example \cite{Feng} propose (in their case, they partition the network based on a hierarchical structure). As we depart from the idea of a complete graph to begin with, partitioning the graph into cliques becomes a \emph{clique problem}. Given a number of cliques $j$, the number of nodes $n$ and the number of members of clique $i$ $n_i$ (where $n_j=n-\sum_{i=1}^{j-1}n_i$), the variance is always minimized as:
$$
\frac{\partial\text{Var}[\xi']}{\partial n_k}=\frac{1}{n_j^2} - \frac{1}{n_k^2} = 0
$$
implying $n_i=\tilde{n} \forall i$ and $\tilde{n}=\frac{n}{j}$. This in turn means that:
$$
\text{Var}[\xi']=\text{Var}[X]\left(\frac{1}{n}\frac{j}{\frac{n}{j}}\right)=\text{Var}[X]\left(\frac{j}{n}\right)^2
$$
So the variance decrease is increasing in $n$ as in standard partitioning problems, but decreasing in $j$ up to the extreme case where $j=n$ (ie. every node is its own clique). Even when abstracting from the fact that the set of all cliques $\mathcal{C}$ in graph $\mathcal{G}$ might not contain ideally sized cliques, the problem of finding a clique cover that minimizes $j$ becomes a \emph{minimal clique problem} and this was shown to be NP-complete by \cite{karp1972reducibility}. Optimizing both clique size and clique number becomes computationally infeasible even with small graphs. 

So unless a very specific network structure (such as a complete graph) is given or a separation into cliques can be translated into a hierarchical structure as in \cite{feng_2020_flood}, working on cliques might present attractive mathematical properties but does not seem practical. Instead, one could consider working with other graph partitioning methods. In the i.i.d. case this might not lead to desirable results, but in the case where there is heterogeneity present in the graph but less so in a densely connected subgroup, there might be some value. For example, if the network considered is of the social kind, densely connected subgroups can represent family or neighbourly ties, which would be interesting  (\cite{bridge2002neighbourhood} or \cite{hanneman2005introduction}) for life-insurance. Subgroup finding on graphs generally works well and a range of fast algorithms exist (e.g. \cite{clauset2004fastgreedy}) and would allow to tune optimizations such as those considered in the linear program and would also allow to depart from the hypothesis of i.i.d. risks. 

}

\subsection{Self-Contribution and ex-post fairness}

Some self contribution is included in most simulations in Section 4 of the main article. Besides the considerations that followed there, this will also result in outcomes that are \emph{perceived} as fair ex-post. We will illustrate the use of the self-contribution layer with an example here. For that consider the three networks depicted in Figure \ref{fig:5_insured_BCE}. Assume there exists a deductible deductible of $100$ and the edges between the nodes are the reciprocal commitments between the nodes. 

\begin{figure}[!ht]
    \centering
    \input{tikz/dessin12.tex}

\vspace{.25cm}

    \caption{Network with four nodes $\{$A,B,C,D$\}$ and three configurations for nodes claiming a loss: nodes \includegraphics[height=.25cm]{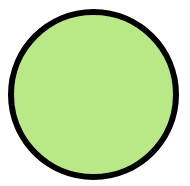} will claim no loss, but will contribute to loss of connections (if any) \includegraphics[height=.25cm]{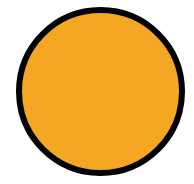} that will claim a loss.}
    \label{fig:5_insured_BCE}
\end{figure}

Here, degrees are $\boldsymbol{d}=(2,2,2,2)$ on the left, and $\boldsymbol{d}=(3,2,2,1)$ for the other two, so that $\overline{d}=2$, and hence we can set the upper bound of the contributions $\gamma$ similarly as in the article to $50$. In Table \ref{tab:ex12}, we have individual losses (denoted $Y$) for two policyholders that are depicted in orange. $X_i$ would be the self contribution for policyholder if there were no risk sharing with friends. $\xi_i$ is the loss of policyholder $i$ in each scenario. Since the sum of $\xi_i$'s equal the sum of $X_i$'s, we consider here some risk sharing, and observe that in all cases, $\text{Var}[\xi_I]\leq \text{Var}[X_I]$. But if we consider the actual payments between the nodes, it becomes clear the ex-post this outcome might not seem fair. If the realized cost is shared pro-rata between the connected nodes, the last line in Table \ref{tab:ex12} depicts the payments. For example, in case a) the nodes that claimed a loss actually do not need to contribute but nodes that do not must reimburse the costs. To avoid such a situation, a first layer of self-contribution can be introduced that will be activate \emph{before} and connected nodes need to contribute anything to the claim.

\begin{table}[]
    \centering
    \begin{tabular}{l|cccc|}
    &\multicolumn{4}{c}{(a)}\\
    
        ~~$i$ & A & B & C & D   \\ \hline
        $d_i$ & 2 & 2 & 2 & 2 \\
        $Z_i$ & 0 & 1 & 0 & 1  \\
        $Y_i$ &  & 200 &  & 60  \\
        $\min\{Y_i,s\}$ & & 100 & & 60\\
        $X_i$ & 0 & 100 & 0 & 60 \\
        $C_i$ & 0 & 50 & 0& 30 \\
        $\xi_i$ & 80 & 0 & 80 & 0 
    \end{tabular}~~\begin{tabular}{|cccc|}
    \multicolumn{4}{c}{(b)}\\
         A & B & C & D   \\ \hline
         3 & 2 & 2 & 1  \\
         0 & 1 & 0 & 1  \\
          & 200 &  & 60  \\
          &100 & & 60\\
         0 & 100 & 0 & 60 \\
         0 & 50 & 0& 50 \\
         100 & 0 & 50 & 10 
    \end{tabular}~~\begin{tabular}{|cccc|}
    \multicolumn{4}{c}{(c)}\\
         A & B & C & D   \\ \hline
         3 & 2 & 2 & 1  \\
         0 & 1 & 0 & 1  \\
         60 & 200 &  &   \\
         60 & 100 & & \\
         60 & 100 & 0 & 0 \\
         20 & 50 & 0& 0 \\
         50 & 20 & 70 & 20 
    \end{tabular}
    
    \vspace{.25cm}
    
    \caption{Scenarios of Figure \ref{fig:5_insured_BCE}, with deductible $s=100$ and (maximal) contribution $\gamma=50$. The average of $x$'s is $\overline{x}=\overline{\xi}=40$, empirical standard deviation of $x$'s  is here proportional to $\sqrt{7200}$, while for $\xi_{(\text{a})}$ it is $\sqrt{6400}$, $\sqrt{6200}$ for $\xi_{(\text{b})}$ and finally $\sqrt{1800}$ for $\xi_{(\text{c})}$.}
    \label{tab:ex12}
\end{table}

% \clearpage
% \bibliography{bibliography}

\end{document}

%% file: tikz/graph_illustration.tex
\begin{tikzpicture}[x=0.75pt,y=0.75pt,yscale=-1,xscale=1]
%uncomment if require: \path (0,300); %set diagram left start at 0, and has height of 300

%Shape: Circle [id:dp686341804548396] 
\draw  [fill={rgb, 255:red, 0; green, 0; blue, 0 }  ,fill opacity=1 ] (47.82,69.52) .. controls (47.82,64.83) and (51.63,61.02) .. (56.32,61.02) .. controls (61.01,61.02) and (64.82,64.83) .. (64.82,69.52) .. controls (64.82,74.22) and (61.01,78.02) .. (56.32,78.02) .. controls (51.63,78.02) and (47.82,74.22) .. (47.82,69.52) -- cycle ;
%Shape: Circle [id:dp6079306272501404] 
\draw  [fill={rgb, 255:red, 0; green, 0; blue, 0 }  ,fill opacity=1 ] (66.18,124.74) .. controls (66.18,120.05) and (69.98,116.24) .. (74.68,116.24) .. controls (79.37,116.24) and (83.18,120.05) .. (83.18,124.74) .. controls (83.18,129.44) and (79.37,133.24) .. (74.68,133.24) .. controls (69.98,133.24) and (66.18,129.44) .. (66.18,124.74) -- cycle ;
%Shape: Circle [id:dp6143383619982439] 
\draw  [fill={rgb, 255:red, 0; green, 0; blue, 0 }  ,fill opacity=1 ] (141.97,68.88) .. controls (141.97,64.19) and (145.78,60.38) .. (150.47,60.38) .. controls (155.17,60.38) and (158.97,64.19) .. (158.97,68.88) .. controls (158.97,73.58) and (155.17,77.38) .. (150.47,77.38) .. controls (145.78,77.38) and (141.97,73.58) .. (141.97,68.88) -- cycle ;
%Shape: Circle [id:dp5096217233190674] 
\draw  [fill={rgb, 255:red, 0; green, 0; blue, 0 }  ,fill opacity=1 ] (124.37,124.35) .. controls (124.37,119.65) and (128.17,115.85) .. (132.87,115.85) .. controls (137.56,115.85) and (141.37,119.65) .. (141.37,124.35) .. controls (141.37,129.04) and (137.56,132.85) .. (132.87,132.85) .. controls (128.17,132.85) and (124.37,129.04) .. (124.37,124.35) -- cycle ;
%Shape: Circle [id:dp7006607886729801] 
\draw  [fill={rgb, 255:red, 0; green, 0; blue, 0 }  ,fill opacity=1 ] (94.66,35) .. controls (94.66,30.31) and (98.47,26.5) .. (103.16,26.5) .. controls (107.86,26.5) and (111.66,30.31) .. (111.66,35) .. controls (111.66,39.7) and (107.86,43.5) .. (103.16,43.5) .. controls (98.47,43.5) and (94.66,39.7) .. (94.66,35) -- cycle ;
%Shape: Regular Polygon [id:dp33573225665778983] 
\draw   (150.47,68.88) -- (132.87,124.35) -- (74.68,124.74) -- (56.32,69.52) -- (103.16,35) -- cycle ;
%Straight Lines [id:da05181262167570133] 
\draw    (103.16,35) -- (74.68,124.74) ;
%Straight Lines [id:da5560787616207097] 
\draw    (103.16,35) -- (132.87,124.35) ;
%Straight Lines [id:da2676383585871247] 
\draw    (141.97,68.88) -- (56.32,69.52) ;
%Straight Lines [id:da9884846041392155] 
\draw    (150.47,68.88) -- (74.68,124.74) ;
%Straight Lines [id:da9560553468111574] 
\draw    (56.32,69.52) -- (132.87,124.35) ;
%Shape: Circle [id:dp10935199882096303] 
\draw  [fill={rgb, 255:red, 0; green, 0; blue, 0 }  ,fill opacity=1 ] (221.66,38) .. controls (221.66,33.31) and (225.47,29.5) .. (230.16,29.5) .. controls (234.86,29.5) and (238.66,33.31) .. (238.66,38) .. controls (238.66,42.7) and (234.86,46.5) .. (230.16,46.5) .. controls (225.47,46.5) and (221.66,42.7) .. (221.66,38) -- cycle ;
%Shape: Circle [id:dp6326124310928694] 
\draw  [fill={rgb, 255:red, 0; green, 0; blue, 0 }  ,fill opacity=1 ] (210.66,80) .. controls (210.66,75.31) and (214.47,71.5) .. (219.16,71.5) .. controls (223.86,71.5) and (227.66,75.31) .. (227.66,80) .. controls (227.66,84.7) and (223.86,88.5) .. (219.16,88.5) .. controls (214.47,88.5) and (210.66,84.7) .. (210.66,80) -- cycle ;
%Shape: Circle [id:dp46718043056965763] 
\draw  [fill={rgb, 255:red, 0; green, 0; blue, 0 }  ,fill opacity=1 ] (264.66,85) .. controls (264.66,80.31) and (268.47,76.5) .. (273.16,76.5) .. controls (277.86,76.5) and (281.66,80.31) .. (281.66,85) .. controls (281.66,89.7) and (277.86,93.5) .. (273.16,93.5) .. controls (268.47,93.5) and (264.66,89.7) .. (264.66,85) -- cycle ;
%Shape: Circle [id:dp7142856446888453] 
\draw  [fill={rgb, 255:red, 0; green, 0; blue, 0 }  ,fill opacity=1 ] (257.66,128) .. controls (257.66,123.31) and (261.47,119.5) .. (266.16,119.5) .. controls (270.86,119.5) and (274.66,123.31) .. (274.66,128) .. controls (274.66,132.7) and (270.86,136.5) .. (266.16,136.5) .. controls (261.47,136.5) and (257.66,132.7) .. (257.66,128) -- cycle ;
%Shape: Circle [id:dp9089248763228164] 
\draw  [fill={rgb, 255:red, 0; green, 0; blue, 0 }  ,fill opacity=1 ] (329.66,88.5) .. controls (329.66,83.81) and (333.47,80) .. (338.16,80) .. controls (342.86,80) and (346.66,83.81) .. (346.66,88.5) .. controls (346.66,93.2) and (342.86,97) .. (338.16,97) .. controls (333.47,97) and (329.66,93.2) .. (329.66,88.5) -- cycle ;
%Straight Lines [id:da34359798568781175] 
\draw    (230.16,38) -- (219.16,80) ;
%Straight Lines [id:da14799918499791276] 
\draw    (219.16,80) -- (273.16,85) ;
%Straight Lines [id:da4950805770465967] 
\draw    (273.16,85) -- (266.16,128) ;
%Straight Lines [id:da8952890992275409] 
\draw    (273.16,85) -- (338.16,88.5) ;
%Straight Lines [id:da3751766733952602] 
\draw    (230.16,38) -- (273.16,85) ;
%Shape: Circle [id:dp720767181943264] 
\draw  [fill={rgb, 255:red, 0; green, 0; blue, 0 }  ,fill opacity=1 ] (375.82,73.52) .. controls (375.82,68.83) and (379.63,65.02) .. (384.32,65.02) .. controls (389.01,65.02) and (392.82,68.83) .. (392.82,73.52) .. controls (392.82,78.22) and (389.01,82.02) .. (384.32,82.02) .. controls (379.63,82.02) and (375.82,78.22) .. (375.82,73.52) -- cycle ;
%Shape: Circle [id:dp2796594989176533] 
\draw  [fill={rgb, 255:red, 0; green, 0; blue, 0 }  ,fill opacity=1 ] (394.18,128.74) .. controls (394.18,124.05) and (397.98,120.24) .. (402.68,120.24) .. controls (407.37,120.24) and (411.18,124.05) .. (411.18,128.74) .. controls (411.18,133.44) and (407.37,137.24) .. (402.68,137.24) .. controls (397.98,137.24) and (394.18,133.44) .. (394.18,128.74) -- cycle ;
%Shape: Circle [id:dp6003846408007134] 
\draw  [fill={rgb, 255:red, 0; green, 0; blue, 0 }  ,fill opacity=1 ] (469.97,72.88) .. controls (469.97,68.19) and (473.78,64.38) .. (478.47,64.38) .. controls (483.17,64.38) and (486.97,68.19) .. (486.97,72.88) .. controls (486.97,77.58) and (483.17,81.38) .. (478.47,81.38) .. controls (473.78,81.38) and (469.97,77.58) .. (469.97,72.88) -- cycle ;
%Shape: Circle [id:dp2199829697887017] 
\draw  [fill={rgb, 255:red, 0; green, 0; blue, 0 }  ,fill opacity=1 ] (452.37,128.35) .. controls (452.37,123.65) and (456.17,119.85) .. (460.87,119.85) .. controls (465.56,119.85) and (469.37,123.65) .. (469.37,128.35) .. controls (469.37,133.04) and (465.56,136.85) .. (460.87,136.85) .. controls (456.17,136.85) and (452.37,133.04) .. (452.37,128.35) -- cycle ;
%Shape: Circle [id:dp7258659713844067] 
\draw  [fill={rgb, 255:red, 0; green, 0; blue, 0 }  ,fill opacity=1 ] (422.66,39) .. controls (422.66,34.31) and (426.47,30.5) .. (431.16,30.5) .. controls (435.86,30.5) and (439.66,34.31) .. (439.66,39) .. controls (439.66,43.7) and (435.86,47.5) .. (431.16,47.5) .. controls (426.47,47.5) and (422.66,43.7) .. (422.66,39) -- cycle ;
%Shape: Circle [id:dp511291770101842] 
\draw  [fill={rgb, 255:red, 0; green, 0; blue, 0 }  ,fill opacity=1 ] (423.66,87) .. controls (423.66,82.31) and (427.47,78.5) .. (432.16,78.5) .. controls (436.86,78.5) and (440.66,82.31) .. (440.66,87) .. controls (440.66,91.7) and (436.86,95.5) .. (432.16,95.5) .. controls (427.47,95.5) and (423.66,91.7) .. (423.66,87) -- cycle ;
%Straight Lines [id:da01770945157188808] 
\draw    (431.16,39) -- (432.16,87) ;
%Straight Lines [id:da6492816605885285] 
\draw    (478.47,72.88) -- (432.16,87) ;
%Straight Lines [id:da8584957572288461] 
\draw    (460.87,128.35) -- (432.16,87) ;
%Straight Lines [id:da6192340494828981] 
\draw    (402.68,128.74) -- (432.16,87) ;
%Straight Lines [id:da439929053072994] 
\draw    (384.32,73.52) -- (432.16,87) ;
%Shape: Circle [id:dp8092175280706155] 
\draw  [fill={rgb, 255:red, 0; green, 0; blue, 0 }  ,fill opacity=1 ] (254.66,16) .. controls (254.66,11.31) and (258.47,7.5) .. (263.16,7.5) .. controls (267.86,7.5) and (271.66,11.31) .. (271.66,16) .. controls (271.66,20.7) and (267.86,24.5) .. (263.16,24.5) .. controls (258.47,24.5) and (254.66,20.7) .. (254.66,16) -- cycle ;
%Straight Lines [id:da34691719296263224] 
\draw    (263.16,16) -- (230.16,38) ;

\end{tikzpicture}

%% file: tikz/info_p2p_scheme.tex
\begin{tikzpicture}[x=0.75pt,y=0.75pt,yscale=-1,xscale=1]
%uncomment if require: \path (0,454); %set diagram left start at 0, and has height of 454

%Straight Lines [id:da20512750599232943] 
\draw  [dash pattern={on 0.84pt off 2.51pt}]  (270,61.5) -- (270,169.5) ;
%Straight Lines [id:da7859891719643155] 
\draw    (29.5,30) -- (554.5,29) ;
\draw [shift={(556.5,29)}, rotate = 179.89] [color={rgb, 255:red, 0; green, 0; blue, 0 }  ][line width=0.75]    (10.93,-3.29) .. controls (6.95,-1.4) and (3.31,-0.3) .. (0,0) .. controls (3.31,0.3) and (6.95,1.4) .. (10.93,3.29)   ;
%Straight Lines [id:da7706167083329676] 
\draw    (29.5,22) -- (29.5,40) ;
%Shape: Rectangle [id:dp04164123853970292] 
\draw  [fill={rgb, 255:red, 190; green, 190; blue, 190 }  ,fill opacity=1 ] (29.5,77) -- (269,77) -- (269,111) -- (29.5,111) -- cycle ;
%Shape: Rectangle [id:dp14684103917007585] 
\draw  [ pattern color={rgb, 255:red, 0; green, 0; blue, 0}] (269.5,77) -- (449.5,77) -- (449.5,111) -- (269.5,111) -- cycle ;
%Shape: Rectangle [id:dp06081153484240098] 
\draw  [fill={rgb, 255:red, 190; green, 190; blue, 190 }  ,fill opacity=1 ] (449,77) -- (551.5,77) -- (551.5,111) -- (449,111) -- cycle ;
%Shape: Rectangle [id:dp13227212981022163] 
\draw  [fill={rgb, 255:red, 190; green, 190; blue, 190 }  ,fill opacity=1 ] (29.5,125) -- (131.5,125) -- (131.5,159) -- (29.5,159) -- cycle ;
%Shape: Rectangle [id:dp7316361167492595] 
\draw  [ pattern color={rgb, 255:red, 0; green, 0; blue, 0}] (268,125) -- (449,125) -- (449,159) -- (268,159) -- cycle ;
%Shape: Rectangle [id:dp6440582566162485] 
\draw  [fill={rgb, 255:red, 190; green, 190; blue, 190 }  ,fill opacity=1 ] (449,125) -- (551.5,125) -- (551.5,159) -- (449,159) -- cycle ;
%Straight Lines [id:da28744580233133066] 
\draw    (472.5,54) -- (450.5,54) ;
\draw [shift={(448.5,54)}, rotate = 360] [color={rgb, 255:red, 0; green, 0; blue, 0 }  ][line width=0.75]    (10.93,-3.29) .. controls (6.95,-1.4) and (3.31,-0.3) .. (0,0) .. controls (3.31,0.3) and (6.95,1.4) .. (10.93,3.29)   ;
%Straight Lines [id:da21596315996894822] 
\draw    (535.5,53) -- (548.5,53) ;
\draw [shift={(550.5,53)}, rotate = 180] [color={rgb, 255:red, 0; green, 0; blue, 0 }  ][line width=0.75]    (10.93,-3.29) .. controls (6.95,-1.4) and (3.31,-0.3) .. (0,0) .. controls (3.31,0.3) and (6.95,1.4) .. (10.93,3.29)   ;
%Shape: Rectangle [id:dp562965027407859] 
\draw  [fill={rgb, 255:red, 74; green, 74; blue, 74 }  ,fill opacity=1 ] (132.5,125) -- (270,125) -- (270,159) -- (132.5,159) -- cycle ;
%Straight Lines [id:da5620785716437715] 
\draw  [dash pattern={on 0.84pt off 2.51pt}]  (448.63,60.5) -- (449.31,178) ;
%Straight Lines [id:da17293470739480998] 
\draw    (29,60.5) -- (270,60.5) ;
%Straight Lines [id:da08895741778340605] 
\draw    (29.5,52) -- (29.5,70) ;
%Straight Lines [id:da22735174361666566] 
\draw    (270,51) -- (270,69) ;

% Text Node
\draw (256,9) node [anchor=north west][inner sep=0.75pt]   [align=left] {claim amount};
% Text Node
\draw (2,84.4) node [anchor=north west][inner sep=0.75pt]    {$( a)$};
% Text Node
\draw (1,132.4) node [anchor=north west][inner sep=0.75pt]    {$( b)$};
% Text Node
\draw (104,86) node [anchor=north west][inner sep=0.75pt]  [color={rgb, 255:red, 255; green, 255; blue, 255 }  ,opacity=1 ] [align=left] {self insurance};
% Text Node
\draw (293,85) node [anchor=north west][inner sep=0.75pt][color={rgb, 255:red, 255; green, 255; blue, 255 }  ,opacity=1 ]   [align=left] {traditional insurance};
% Text Node
\draw (464,85) node [anchor=north west][inner sep=0.75pt]  [color={rgb, 255:red, 255; green, 255; blue, 255 }  ,opacity=1 ] [align=left] {reinsurance};
% Text Node
\draw (33,134) node [anchor=north west][inner sep=0.75pt]  [color={rgb, 255:red, 255; green, 255; blue, 255 }  ,opacity=1 ] [align=left] {self insurance};
% Text Node
\draw (293,133) node [anchor=north west][inner sep=0.75pt][color={rgb, 255:red, 255; green, 255; blue, 255 }  ,opacity=1 ] [align=left] {traditional insurance};
% Text Node
\draw (465,133) node [anchor=north west][inner sep=0.75pt]  [color={rgb, 255:red, 255; green, 255; blue, 255 }  ,opacity=1 ] [align=left] {reinsurance};
% Text Node
\draw (119,39) node [anchor=north west][inner sep=0.75pt]   [align=left] {deductible};
% Text Node
\draw (470,34) node [anchor=north west][inner sep=0.75pt]   [align=left] {\begin{minipage}[lt]{44.68pt}\setlength\topsep{0pt}
\begin{center}
\textit{excess}
\end{center}
\textit{coverage}
\end{minipage}};
% Text Node
\draw (188,133) node [anchor=north west][inner sep=0.75pt]  [color={rgb, 255:red, 255; green, 255; blue, 255 }  ,opacity=1 ] [align=left] {\textcolor[rgb]{1,1,1}{p2p}};

\end{tikzpicture}

%% file: tikz/figure_pyramide.tex
% Pattern Info
 
\tikzset{
pattern size/.store in=\mcSize, 
pattern size = 5pt,
pattern thickness/.store in=\mcThickness, 
pattern thickness = 0.3pt,
pattern radius/.store in=\mcRadius, 
pattern radius = 1pt}
\makeatletter
\pgfutil@ifundefined{pgf@pattern@name@_muhlx63dh}{
\makeatletter
\pgfdeclarepatternformonly[\mcRadius,\mcThickness,\mcSize]{_muhlx63dh}
{\pgfpoint{-0.5*\mcSize}{-0.5*\mcSize}}
{\pgfpoint{0.5*\mcSize}{0.5*\mcSize}}
{\pgfpoint{\mcSize}{\mcSize}}
{
\pgfsetcolor{\tikz@pattern@color}
\pgfsetlinewidth{\mcThickness}
\pgfpathcircle\pgfpointorigin{\mcRadius}
\pgfusepath{stroke}
}}
\makeatother

% Pattern Info
 
\tikzset{
pattern size/.store in=\mcSize, 
pattern size = 5pt,
pattern thickness/.store in=\mcThickness, 
pattern thickness = 0.3pt,
pattern radius/.store in=\mcRadius, 
pattern radius = 1pt}
\makeatletter
\pgfutil@ifundefined{pgf@pattern@name@_aqbesnm4p}{
\pgfdeclarepatternformonly[\mcThickness,\mcSize]{_aqbesnm4p}
{\pgfqpoint{0pt}{0pt}}
{\pgfpoint{\mcSize+\mcThickness}{\mcSize+\mcThickness}}
{\pgfpoint{\mcSize}{\mcSize}}
{
\pgfsetcolor{\tikz@pattern@color}
\pgfsetlinewidth{\mcThickness}
\pgfpathmoveto{\pgfqpoint{0pt}{0pt}}
\pgfpathlineto{\pgfpoint{\mcSize+\mcThickness}{\mcSize+\mcThickness}}
\pgfusepath{stroke}
}}
\makeatother

% Pattern Info
 
\tikzset{
pattern size/.store in=\mcSize, 
pattern size = 5pt,
pattern thickness/.store in=\mcThickness, 
pattern thickness = 0.3pt,
pattern radius/.store in=\mcRadius, 
pattern radius = 1pt}
\makeatletter
\pgfutil@ifundefined{pgf@pattern@name@_wwrm0e2rc}{
\makeatletter
\pgfdeclarepatternformonly[\mcRadius,\mcThickness,\mcSize]{_wwrm0e2rc}
{\pgfpoint{-0.5*\mcSize}{-0.5*\mcSize}}
{\pgfpoint{0.5*\mcSize}{0.5*\mcSize}}
{\pgfpoint{\mcSize}{\mcSize}}
{
\pgfsetcolor{\tikz@pattern@color}
\pgfsetlinewidth{\mcThickness}
\pgfpathcircle\pgfpointorigin{\mcRadius}
\pgfusepath{stroke}
}}
\makeatother

% Pattern Info
 
\tikzset{
pattern size/.store in=\mcSize, 
pattern size = 5pt,
pattern thickness/.store in=\mcThickness, 
pattern thickness = 0.3pt,
pattern radius/.store in=\mcRadius, 
pattern radius = 1pt}
\makeatletter
\pgfutil@ifundefined{pgf@pattern@name@_invdx69uf}{
\pgfdeclarepatternformonly[\mcThickness,\mcSize]{_invdx69uf}
{\pgfqpoint{0pt}{0pt}}
{\pgfpoint{\mcSize+\mcThickness}{\mcSize+\mcThickness}}
{\pgfpoint{\mcSize}{\mcSize}}
{
\pgfsetcolor{\tikz@pattern@color}
\pgfsetlinewidth{\mcThickness}
\pgfpathmoveto{\pgfqpoint{0pt}{0pt}}
\pgfpathlineto{\pgfpoint{\mcSize+\mcThickness}{\mcSize+\mcThickness}}
\pgfusepath{stroke}
}}
\makeatother

% Pattern Info
 
\tikzset{
pattern size/.store in=\mcSize, 
pattern size = 5pt,
pattern thickness/.store in=\mcThickness, 
pattern thickness = 0.3pt,
pattern radius/.store in=\mcRadius, 
pattern radius = 1pt}
\makeatletter
\pgfutil@ifundefined{pgf@pattern@name@_gzkd5s1xd}{
\pgfdeclarepatternformonly[\mcThickness,\mcSize]{_gzkd5s1xd}
{\pgfqpoint{0pt}{0pt}}
{\pgfpoint{\mcSize+\mcThickness}{\mcSize+\mcThickness}}
{\pgfpoint{\mcSize}{\mcSize}}
{
\pgfsetcolor{\tikz@pattern@color}
\pgfsetlinewidth{\mcThickness}
\pgfpathmoveto{\pgfqpoint{0pt}{0pt}}
\pgfpathlineto{\pgfpoint{\mcSize+\mcThickness}{\mcSize+\mcThickness}}
\pgfusepath{stroke}
}}
\makeatother

% Pattern Info
 
\tikzset{
pattern size/.store in=\mcSize, 
pattern size = 5pt,
pattern thickness/.store in=\mcThickness, 
pattern thickness = 0.3pt,
pattern radius/.store in=\mcRadius, 
pattern radius = 1pt}
\makeatletter
\pgfutil@ifundefined{pgf@pattern@name@_pb8qg8ygw}{
\makeatletter
\pgfdeclarepatternformonly[\mcRadius,\mcThickness,\mcSize]{_pb8qg8ygw}
{\pgfpoint{-0.5*\mcSize}{-0.5*\mcSize}}
{\pgfpoint{0.5*\mcSize}{0.5*\mcSize}}
{\pgfpoint{\mcSize}{\mcSize}}
{
\pgfsetcolor{\tikz@pattern@color}
\pgfsetlinewidth{\mcThickness}
\pgfpathcircle\pgfpointorigin{\mcRadius}
\pgfusepath{stroke}
}}
\makeatother
\tikzset{every picture/.style={line width=0.75pt}} %set default line width to 0.75pt        

\begin{tikzpicture}[x=0.75pt,y=0.75pt,yscale=-1,xscale=1]
%uncomment if require: \path (0,420); %set diagram left start at 0, and has height of 420

%Shape: Right Triangle [id:dp5491391010154183] 
\draw  [color={rgb, 255:red, 0; green, 0; blue, 0 }  ,draw opacity=1 ][pattern=_muhlx63dh,pattern size=6pt,pattern thickness=0.75pt,pattern radius=0.75pt, pattern color={rgb, 255:red, 0; green, 0; blue, 0}] (483,42) -- (399.5,134) -- (483,134) -- cycle ;
%Straight Lines [id:da3214182619272712] 
\draw    (316,227) -- (563.5,227) ;
\draw [shift={(565.5,227)}, rotate = 180] [color={rgb, 255:red, 0; green, 0; blue, 0 }  ][line width=0.75]    (10.93,-3.29) .. controls (6.95,-1.4) and (3.31,-0.3) .. (0,0) .. controls (3.31,0.3) and (6.95,1.4) .. (10.93,3.29)   ;
%Straight Lines [id:da40912792906721807] 
\draw    (316,226) -- (314.51,3) ;
\draw [shift={(314.5,1)}, rotate = 89.62] [color={rgb, 255:red, 0; green, 0; blue, 0 }  ][line width=0.75]    (10.93,-3.29) .. controls (6.95,-1.4) and (3.31,-0.3) .. (0,0) .. controls (3.31,0.3) and (6.95,1.4) .. (10.93,3.29)   ;
%Shape: Right Triangle [id:dp844157995560223] 
\draw  [draw opacity=0][fill={rgb, 255:red, 74; green, 74; blue, 74 }  ,fill opacity=1 ] (399.5,134) -- (316,226) -- (399.5,226) -- cycle ;
%Shape: Rectangle [id:dp9507864388012495] 
\draw  [draw opacity=0][fill={rgb, 255:red, 74; green, 74; blue, 74 }  ,fill opacity=1 ] (399.5,134) -- (543.5,134) -- (543.5,226) -- (399.5,226) -- cycle ;
%Shape: Right Triangle [id:dp6355129318582893] 
\draw  [color={rgb, 255:red, 0; green, 0; blue, 0 }  ,draw opacity=1 ][pattern=_aqbesnm4p,pattern size=6pt,pattern thickness=0.75pt,pattern radius=0pt, pattern color={rgb, 255:red, 0; green, 0; blue, 0}] (483,94) -- (399.5,134) -- (483,134) -- cycle ;
%Shape: Right Triangle [id:dp841694709549772] 
\draw  [color={rgb, 255:red, 0; green, 0; blue, 0 }  ,draw opacity=1 ][fill={rgb, 255:red, 222; green, 222; blue, 222 }  ,fill opacity=1 ] (517.5,3) -- (483,42) -- (517.5,42) -- cycle ;
%Shape: Rectangle [id:dp9441848418202281] 
\draw  [color={rgb, 255:red, 0; green, 0; blue, 0 }  ,draw opacity=1 ][pattern=_wwrm0e2rc,pattern size=6pt,pattern thickness=0.75pt,pattern radius=0.75pt, pattern color={rgb, 255:red, 0; green, 0; blue, 0}] (483,42) -- (543.5,42) -- (543.5,94) -- (483,94) -- cycle ;
%Shape: Rectangle [id:dp5066676919736814] 
\draw  [color={rgb, 255:red, 0; green, 0; blue, 0 }  ,draw opacity=1 ][pattern=_invdx69uf,pattern size=6pt,pattern thickness=0.75pt,pattern radius=0pt, pattern color={rgb, 255:red, 0; green, 0; blue, 0}] (483,94) -- (544,94) -- (544,134.5) -- (483,134.5) -- cycle ;
%Shape: Rectangle [id:dp535415169350355] 
\draw  [color={rgb, 255:red, 0; green, 0; blue, 0 }  ,draw opacity=1 ][fill={rgb, 255:red, 222; green, 222; blue, 222 }  ,fill opacity=1 ] (517.5,3) -- (543.5,3) -- (543.5,42) -- (517.5,42) -- cycle ;
%Straight Lines [id:da24475422038128603] 
\draw  [dash pattern={on 4.5pt off 4.5pt}]  (399.5,5) -- (399.5,226) ;
%Straight Lines [id:da8365266777655516] 
\draw  [dash pattern={on 4.5pt off 4.5pt}]  (483.5,5) -- (483.5,228) ;
%Straight Lines [id:da895465857210402] 
\draw    (517.5,4) -- (518.5,225) ;
%Straight Lines [id:da2428807081717781] 
\draw    (401.5,10.02) -- (480,10.98) ;
\draw [shift={(482,11)}, rotate = 180.69] [color={rgb, 255:red, 0; green, 0; blue, 0 }  ][line width=0.75]    (7.65,-3.43) .. controls (4.86,-1.61) and (2.31,-0.47) .. (0,0) .. controls (2.31,0.47) and (4.86,1.61) .. (7.65,3.43)   ;
\draw [shift={(399.5,10)}, rotate = 0.69] [color={rgb, 255:red, 0; green, 0; blue, 0 }  ][line width=0.75]    (7.65,-3.43) .. controls (4.86,-1.61) and (2.31,-0.47) .. (0,0) .. controls (2.31,0.47) and (4.86,1.61) .. (7.65,3.43)   ;
%Straight Lines [id:da013175223295107452] 
\draw    (467.5,61) -- (467.5,100) ;
\draw [shift={(467.5,103)}, rotate = 270] [fill={rgb, 255:red, 0; green, 0; blue, 0 }  ][line width=0.08]  [draw opacity=0] (8.93,-4.29) -- (0,0) -- (8.93,4.29) -- cycle    ;
\draw [shift={(467.5,58)}, rotate = 90] [fill={rgb, 255:red, 0; green, 0; blue, 0 }  ][line width=0.08]  [draw opacity=0] (8.93,-4.29) -- (0,0) -- (8.93,4.29) -- cycle    ;
%Straight Lines [id:da39136523871310436] 
\draw    (569.53,141) -- (570.47,222) ;
\draw [shift={(570.5,225)}, rotate = 269.34] [fill={rgb, 255:red, 0; green, 0; blue, 0 }  ][line width=0.08]  [draw opacity=0] (8.93,-4.29) -- (0,0) -- (8.93,4.29) -- cycle    ;
\draw [shift={(569.5,138)}, rotate = 89.34] [fill={rgb, 255:red, 0; green, 0; blue, 0 }  ][line width=0.08]  [draw opacity=0] (8.93,-4.29) -- (0,0) -- (8.93,4.29) -- cycle    ;
%Straight Lines [id:da37289276351320244] 
\draw    (569.58,10) -- (570.42,40) ;
\draw [shift={(570.5,43)}, rotate = 268.41] [fill={rgb, 255:red, 0; green, 0; blue, 0 }  ][line width=0.08]  [draw opacity=0] (8.93,-4.29) -- (0,0) -- (8.93,4.29) -- cycle    ;
\draw [shift={(569.5,7)}, rotate = 88.41] [fill={rgb, 255:red, 0; green, 0; blue, 0 }  ][line width=0.08]  [draw opacity=0] (8.93,-4.29) -- (0,0) -- (8.93,4.29) -- cycle    ;
%Shape: Circle [id:dp12270314946438776] 
\draw   (151.8,160) .. controls (151.8,151.16) and (158.96,144) .. (167.8,144) .. controls (176.64,144) and (183.8,151.16) .. (183.8,160) .. controls (183.8,168.84) and (176.64,176) .. (167.8,176) .. controls (158.96,176) and (151.8,168.84) .. (151.8,160) -- cycle ;
%Shape: Circle [id:dp8103472255817286] 
\draw  [fill={rgb, 255:red, 155; green, 155; blue, 155 }  ,fill opacity=1 ] (59.8,131) .. controls (59.8,122.16) and (66.96,115) .. (75.8,115) .. controls (84.64,115) and (91.8,122.16) .. (91.8,131) .. controls (91.8,139.84) and (84.64,147) .. (75.8,147) .. controls (66.96,147) and (59.8,139.84) .. (59.8,131) -- cycle ;
%Shape: Circle [id:dp6133880974320782] 
\draw  [fill={rgb, 255:red, 155; green, 155; blue, 155 }  ,fill opacity=1 ] (141.8,147) .. controls (141.8,138.16) and (148.96,131) .. (157.8,131) .. controls (166.64,131) and (173.8,138.16) .. (173.8,147) .. controls (173.8,155.84) and (166.64,163) .. (157.8,163) .. controls (148.96,163) and (141.8,155.84) .. (141.8,147) -- cycle ;
%Rounded Rect [id:dp8650277326955156] 
\draw  [fill={rgb, 255:red, 222; green, 222; blue, 222 }  ,fill opacity=1 ] (149.8,178.2) .. controls (149.8,175.88) and (151.68,174) .. (154,174) -- (174.6,174) .. controls (176.92,174) and (178.8,175.88) .. (178.8,178.2) -- (178.8,190.8) .. controls (178.8,193.12) and (176.92,195) .. (174.6,195) -- (154,195) .. controls (151.68,195) and (149.8,193.12) .. (149.8,190.8) -- cycle ;
%Rounded Rect [id:dp4243614882427793] 
\draw  [fill={rgb, 255:red, 74; green, 74; blue, 74 }  ,fill opacity=1 ] (179.8,143.2) .. controls (179.8,140.88) and (181.68,139) .. (184,139) -- (204.6,139) .. controls (206.92,139) and (208.8,140.88) .. (208.8,143.2) -- (208.8,155.8) .. controls (208.8,158.12) and (206.92,160) .. (204.6,160) -- (184,160) .. controls (181.68,160) and (179.8,158.12) .. (179.8,155.8) -- cycle ;
%Straight Lines [id:da181369617462938] 
\draw    (91,132) -- (141,148) ;
%Rounded Rect [id:dp37107325913211886] 
\draw  [pattern=_gzkd5s1xd,pattern size=6pt,pattern thickness=0.75pt,pattern radius=0pt, pattern color={rgb, 255:red, 0; green, 0; blue, 0}] (101,134.2) .. controls (101,131.88) and (102.88,130) .. (105.2,130) -- (125.8,130) .. controls (128.12,130) and (130,131.88) .. (130,134.2) -- (130,146.8) .. controls (130,149.12) and (128.12,151) .. (125.8,151) -- (105.2,151) .. controls (102.88,151) and (101,149.12) .. (101,146.8) -- cycle ;
%Shape: Circle [id:dp11338870860282313] 
\draw  [fill={rgb, 255:red, 155; green, 155; blue, 155 }  ,fill opacity=1 ] (162,71) .. controls (162,62.16) and (169.16,55) .. (178,55) .. controls (186.84,55) and (194,62.16) .. (194,71) .. controls (194,79.84) and (186.84,87) .. (178,87) .. controls (169.16,87) and (162,79.84) .. (162,71) -- cycle ;
%Straight Lines [id:da6371729229072176] 
\draw    (163,133) -- (172.5,87) ;
%Rounded Rect [id:dp8658718334978699] 
\draw  [pattern=_pb8qg8ygw,pattern size=6pt,pattern thickness=0.75pt,pattern radius=0.75pt, pattern color={rgb, 255:red, 0; green, 0; blue, 0}] (152,102.2) .. controls (152,99.88) and (153.88,98) .. (156.2,98) -- (176.8,98) .. controls (179.12,98) and (181,99.88) .. (181,102.2) -- (181,114.8) .. controls (181,117.12) and (179.12,119) .. (176.8,119) -- (156.2,119) .. controls (153.88,119) and (152,117.12) .. (152,114.8) -- cycle ;
%Straight Lines [id:da2989087338527414] 
\draw  [dash pattern={on 0.84pt off 2.51pt}]  (81.5,115) -- (85.5,88) ;
%Straight Lines [id:da716095227047407] 
\draw  [dash pattern={on 0.84pt off 2.51pt}]  (162,71) -- (137.5,79) ;
%Straight Lines [id:da6272541709559742] 
\draw  [dash pattern={on 0.84pt off 2.51pt}]  (178,32) -- (178,55) ;
%Straight Lines [id:da9287114115462796] 
\draw  [dash pattern={on 0.84pt off 2.51pt}]  (70,147) -- (52.5,169) ;
%Straight Lines [id:da0764593327208829] 
\draw  [dash pattern={on 0.84pt off 2.51pt}]  (218.5,85) -- (194,71) ;
%Straight Lines [id:da9605671386088506] 
\draw    (467.5,106) -- (467.5,110) -- (467.5,133) ;
\draw [shift={(467.5,136)}, rotate = 270] [fill={rgb, 255:red, 0; green, 0; blue, 0 }  ][line width=0.08]  [draw opacity=0] (8.93,-4.29) -- (0,0) -- (8.93,4.29) -- cycle    ;
\draw [shift={(467.5,103)}, rotate = 90] [fill={rgb, 255:red, 0; green, 0; blue, 0 }  ][line width=0.08]  [draw opacity=0] (8.93,-4.29) -- (0,0) -- (8.93,4.29) -- cycle    ;

% Text Node
\draw (551.8,172) node [anchor=north west][inner sep=0.75pt]   [align=left] {A};
% Text Node
\draw (550.8,17) node [anchor=north west][inner sep=0.75pt]   [align=left] {A};
% Text Node
\draw (550.8,55) node [anchor=north west][inner sep=0.75pt]   [align=left] {C};
% Text Node
\draw (551.8,104) node [anchor=north west][inner sep=0.75pt]   [align=left] {B};
% Text Node
\draw (404.5,-13) node [anchor=north west][inner sep=0.75pt]   [align=left] {reciprocals};
% Text Node
\draw (503,233.4) node [anchor=north west][inner sep=0.75pt]  [font=\normalsize]  {$100$};
% Text Node
\draw (293.55,144.53) node [anchor=north west][inner sep=0.75pt]  [rotate=-269.79] [align=left] {contributions};
% Text Node
\draw (388.5,263) node [anchor=north west][inner sep=0.75pt]   [align=left] {amount of the claim};
% Text Node
\draw (310,233.4) node [anchor=north west][inner sep=0.75pt]    {$0$};
% Text Node
\draw (508.55,147.53) node [anchor=north west][fill=white, inner sep=0.75pt]  [rotate=-269.79] [align=left] { deductible $\displaystyle s$};
% Text Node
\draw (386,233.4) node [anchor=north west][inner sep=0.75pt]  [font=\normalsize]  {$40$};
% Text Node
\draw (474,233.4) node [anchor=north west][inner sep=0.75pt]  [font=\normalsize]  {$90$};
% Text Node
\draw (476.63,133.37) node [anchor=north west][fill=white, inner sep=0.75pt]  [font=\footnotesize,rotate=-270.95]  {$\ 40\%$};
% Text Node
\draw (475.63,94.37) node [anchor=north west][fill=white, inner sep=0.75pt]  [font=\footnotesize,rotate=-270.95]  {$\ 60\%$};
% Text Node
\draw (578.73,232.93) node [anchor=north west][inner sep=0.75pt]  [font=\small,rotate=-270.49] [align=left] {\begin{minipage}[lt]{65.48pt}\setlength\topsep{0pt}
\begin{center}
first layer\\(self insurance)
\end{center}

\end{minipage}};
% Text Node
\draw (580.73,81.93) node [anchor=north west][inner sep=0.75pt]  [font=\small,rotate=-270.49] [align=left] {\begin{minipage}[lt]{59.36pt}\setlength\topsep{0pt}
\begin{center}
residual\\self insurance
\end{center}

\end{minipage}};
% Text Node
\draw (69,124) node [anchor=north west][inner sep=0.75pt]   [align=left] {B};
% Text Node
\draw (151,139) node [anchor=north west][inner sep=0.75pt]   [align=left] {A};
% Text Node
\draw (154,176) node [anchor=north west][inner sep=1.75pt]  [color={rgb, 255:red, 255; green, 255; blue, 255 }  ,opacity=1 ] [align=left] {10 };
% Text Node
\draw (185,141) node [anchor=north west][inner sep=1.75pt]  [color={rgb, 255:red, 255; green, 255; blue, 255 }  ,opacity=1 ] [align=left] {40};
% Text Node
\draw (106.2,132) node [anchor=north west][fill=white, inner sep=1.75pt]   [align=left] {20};
% Text Node
\draw (171,64) node [anchor=north west][inner sep=0.75pt]   [align=left] {C};
% Text Node
\draw (157.2,100) node [anchor=north west][fill=white,inner sep=1.75pt]   [align=left] {30};
% Text Node
\draw (231.55,177.53) node [anchor=north west][inner sep=0.75pt]  [rotate=-269.79] [align=left] {contributions};

\end{tikzpicture}

%% file: tikz/dessin8.tex
\tikzset{
pattern size/.store in=\mcSize, 
pattern size = 5pt,
pattern thickness/.store in=\mcThickness, 
pattern thickness = 0.3pt,
pattern radius/.store in=\mcRadius, 
pattern radius = 1pt}
\makeatletter
\pgfutil@ifundefined{pgf@pattern@name@_955ioluz3}{
\pgfdeclarepatternformonly[\mcThickness,\mcSize]{_955ioluz3}
{\pgfqpoint{0pt}{-\mcThickness}}
{\pgfpoint{\mcSize}{\mcSize}}
{\pgfpoint{\mcSize}{\mcSize}}
{
\pgfsetcolor{\tikz@pattern@color}
\pgfsetlinewidth{\mcThickness}
\pgfpathmoveto{\pgfqpoint{0pt}{\mcSize}}
\pgfpathlineto{\pgfpoint{\mcSize+\mcThickness}{-\mcThickness}}
\pgfusepath{stroke}
}}
\makeatother
\tikzset{every picture/.style={line width=0.75pt}} %set default line width to 0.75pt        

\begin{tikzpicture}[x=0.75pt,y=0.75pt,yscale=-1,xscale=1]
%uncomment if require: \path (0,605); %set diagram left start at 0, and has height of 605

%Shape: Circle [id:dp737263028349751] 
\draw   (141.8,185) .. controls (141.8,176.16) and (148.96,169) .. (157.8,169) .. controls (166.64,169) and (173.8,176.16) .. (173.8,185) .. controls (173.8,193.84) and (166.64,201) .. (157.8,201) .. controls (148.96,201) and (141.8,193.84) .. (141.8,185) -- cycle ;
%Shape: Circle [id:dp42195339977830915] 
\draw   (32.8,170) .. controls (32.8,161.16) and (39.96,154) .. (48.8,154) .. controls (57.64,154) and (64.8,161.16) .. (64.8,170) .. controls (64.8,178.84) and (57.64,186) .. (48.8,186) .. controls (39.96,186) and (32.8,178.84) .. (32.8,170) -- cycle ;
%Shape: Circle [id:dp9602693493583863] 
\draw   (169.8,81) .. controls (169.8,72.16) and (176.96,65) .. (185.8,65) .. controls (194.64,65) and (201.8,72.16) .. (201.8,81) .. controls (201.8,89.84) and (194.64,97) .. (185.8,97) .. controls (176.96,97) and (169.8,89.84) .. (169.8,81) -- cycle ;
%Shape: Circle [id:dp23609896266431818] 
\draw   (55.8,60) .. controls (55.8,51.16) and (62.96,44) .. (71.8,44) .. controls (80.64,44) and (87.8,51.16) .. (87.8,60) .. controls (87.8,68.84) and (80.64,76) .. (71.8,76) .. controls (62.96,76) and (55.8,68.84) .. (55.8,60) -- cycle ;
%Shape: Circle [id:dp4235283908364136] 
\draw  [fill={rgb, 255:red, 188; green, 188; blue, 188 }  ,fill opacity=1 ] (155.8,96) .. controls (155.8,87.16) and (162.96,80) .. (171.8,80) .. controls (180.64,80) and (187.8,87.16) .. (187.8,96) .. controls (187.8,104.84) and (180.64,112) .. (171.8,112) .. controls (162.96,112) and (155.8,104.84) .. (155.8,96) -- cycle ;
%Shape: Circle [id:dp32973230092805816] 
\draw  [fill={rgb, 255:red, 188; green, 188; blue, 188 }  ,fill opacity=1 ] (49.8,156) .. controls (49.8,147.16) and (56.96,140) .. (65.8,140) .. controls (74.64,140) and (81.8,147.16) .. (81.8,156) .. controls (81.8,164.84) and (74.64,172) .. (65.8,172) .. controls (56.96,172) and (49.8,164.84) .. (49.8,156) -- cycle ;
%Straight Lines [id:da16786206652627866] 
\draw    (83.8,82) -- (65.8,140) ;
%Shape: Circle [id:dp6804979819088912] 
\draw  [fill={rgb, 255:red, 188; green, 188; blue, 188 }  ,fill opacity=1 ] (74.8,67) .. controls (74.8,58.16) and (81.96,51) .. (90.8,51) .. controls (99.64,51) and (106.8,58.16) .. (106.8,67) .. controls (106.8,75.84) and (99.64,83) .. (90.8,83) .. controls (81.96,83) and (74.8,75.84) .. (74.8,67) -- cycle ;
%Straight Lines [id:da2797304109266483] 
\draw    (155.8,96) -- (104.8,73) ;
%Shape: Circle [id:dp06851148238813964] 
\draw  [fill={rgb, 255:red, 188; green, 188; blue, 188 }  ,fill opacity=1 ] (131.8,172) .. controls (131.8,163.16) and (138.96,156) .. (147.8,156) .. controls (156.64,156) and (163.8,163.16) .. (163.8,172) .. controls (163.8,180.84) and (156.64,188) .. (147.8,188) .. controls (138.96,188) and (131.8,180.84) .. (131.8,172) -- cycle ;
%Straight Lines [id:da1901144267738153] 
\draw    (81.8,156) -- (131.8,172) ;
%Rounded Rect [id:dp7808174969415405] 
\draw  [fill={rgb, 255:red, 220; green, 220; blue, 220 }  ,fill opacity=1 ] (62.8,99.2) .. controls (62.8,96.88) and (64.68,95) .. (67,95) -- (87.6,95) .. controls (89.92,95) and (91.8,96.88) .. (91.8,99.2) -- (91.8,111.8) .. controls (91.8,114.12) and (89.92,116) .. (87.6,116) -- (67,116) .. controls (64.68,116) and (62.8,114.12) .. (62.8,111.8) -- cycle ;
%Rounded Rect [id:dp5197285945839982] 
\draw  [fill={rgb, 255:red, 220; green, 220; blue, 220 }  ,fill opacity=1 ] (114.8,76.2) .. controls (114.8,73.88) and (116.68,72) .. (119,72) -- (139.6,72) .. controls (141.92,72) and (143.8,73.88) .. (143.8,76.2) -- (143.8,88.8) .. controls (143.8,91.12) and (141.92,93) .. (139.6,93) -- (119,93) .. controls (116.68,93) and (114.8,91.12) .. (114.8,88.8) -- cycle ;
%Rounded Rect [id:dp38632728161098995] 
\draw  [fill={rgb, 255:red, 74; green, 74; blue, 74 }  ,fill opacity=1 ] (39.8,43.2) .. controls (39.8,40.88) and (41.68,39) .. (44,39) -- (64.6,39) .. controls (66.92,39) and (68.8,40.88) .. (68.8,43.2) -- (68.8,55.8) .. controls (68.8,58.12) and (66.92,60) .. (64.6,60) -- (44,60) .. controls (41.68,60) and (39.8,58.12) .. (39.8,55.8) -- cycle ;
%Rounded Rect [id:dp782879374688584] 
\draw  [fill={rgb, 255:red, 74; green, 74; blue, 74 }  ,fill opacity=1 ] (16.8,175.2) .. controls (16.8,172.88) and (18.68,171) .. (21,171) -- (41.6,171) .. controls (43.92,171) and (45.8,172.88) .. (45.8,175.2) -- (45.8,187.8) .. controls (45.8,190.12) and (43.92,192) .. (41.6,192) -- (21,192) .. controls (18.68,192) and (16.8,190.12) .. (16.8,187.8) -- cycle ;
%Rounded Rect [id:dp3879620776819974] 
\draw  [pattern=_955ioluz3,pattern size=6pt,pattern thickness=0.75pt,pattern radius=0pt, pattern color={rgb, 255:red, 0; green, 0; blue, 0}] (137.8,200.2) .. controls (137.8,197.88) and (139.68,196) .. (142,196) -- (162.6,196) .. controls (164.92,196) and (166.8,197.88) .. (166.8,200.2) -- (166.8,212.8) .. controls (166.8,215.12) and (164.92,217) .. (162.6,217) -- (142,217) .. controls (139.68,217) and (137.8,215.12) .. (137.8,212.8) -- cycle ;
%Rounded Rect [id:dp9453988490828278] 
\draw  [fill={rgb, 255:red, 220; green, 220; blue, 220 }  ,fill opacity=1 ] (92.8,158.2) .. controls (92.8,155.88) and (94.68,154) .. (97,154) -- (117.6,154) .. controls (119.92,154) and (121.8,155.88) .. (121.8,158.2) -- (121.8,170.8) .. controls (121.8,173.12) and (119.92,175) .. (117.6,175) -- (97,175) .. controls (94.68,175) and (92.8,173.12) .. (92.8,170.8) -- cycle ;
%Rounded Rect [id:dp12897034749159397] 
\draw  [fill={rgb, 255:red, 74; green, 74; blue, 74 }  ,fill opacity=1 ] (195.8,79.2) .. controls (195.8,76.88) and (197.68,75) .. (200,75) -- (220.6,75) .. controls (222.92,75) and (224.8,76.88) .. (224.8,79.2) -- (224.8,91.8) .. controls (224.8,94.12) and (222.92,96) .. (220.6,96) -- (200,96) .. controls (197.68,96) and (195.8,94.12) .. (195.8,91.8) -- cycle ;
%Rounded Rect [id:dp31850978613396963] 
\draw  [fill={rgb, 255:red, 74; green, 74; blue, 74 }  ,fill opacity=1 ] (167.8,171.2) .. controls (167.8,168.88) and (169.68,167) .. (172,167) -- (192.6,167) .. controls (194.92,167) and (196.8,168.88) .. (196.8,171.2) -- (196.8,183.8) .. controls (196.8,186.12) and (194.92,188) .. (192.6,188) -- (172,188) .. controls (169.68,188) and (167.8,186.12) .. (167.8,183.8) -- cycle ;
%Straight Lines [id:da07828150965171088] 
\draw    (76.8,142) -- (158.8,104) ;
%Rounded Rect [id:dp5083470082040547] 
\draw  [fill={rgb, 255:red, 220; green, 220; blue, 220 }  ,fill opacity=1 ] (104.8,114.2) .. controls (104.8,111.88) and (106.68,110) .. (109,110) -- (129.6,110) .. controls (131.92,110) and (133.8,111.88) .. (133.8,114.2) -- (133.8,126.8) .. controls (133.8,129.12) and (131.92,131) .. (129.6,131) -- (109,131) .. controls (106.68,131) and (104.8,129.12) .. (104.8,126.8) -- cycle ;
%Curve Lines [id:da17625400201352193] 
\draw [color={rgb, 255:red, 74; green, 74; blue, 74 }  ,draw opacity=1 ] [dash pattern={on 4.5pt off 4.5pt}]  (161.5,163) .. controls (248.5,157) and (267.5,119) .. (184.5,105) ;
%Rounded Rect [id:dp7801918278110926] 
\draw  [fill={rgb, 255:red, 235; green, 235; blue, 235 }  ,fill opacity=1 ] (224,133.2) .. controls (224,130.88) and (225.88,129) .. (228.2,129) -- (248.8,129) .. controls (251.12,129) and (253,130.88) .. (253,133.2) -- (253,145.8) .. controls (253,148.12) and (251.12,150) .. (248.8,150) -- (228.2,150) .. controls (225.88,150) and (224,148.12) .. (224,145.8) -- cycle ;

% Text Node
\draw (59.8,147) node [anchor=north west][inner sep=0.75pt]   [align=left] {A};
% Text Node
\draw (84.8,58) node [anchor=north west][inner sep=0.75pt]   [align=left] {B};
% Text Node
\draw (164.8,87) node [anchor=north west][inner sep=0.75pt]   [align=left] {C};
% Text Node
\draw (141.8,163) node [anchor=north west][inner sep=0.75pt]   [align=left] {D};
% Text Node
\draw (69,98) node [anchor=north west][inner sep=0.75pt]   [align=left] {20};
% Text Node
\draw (120,74) node [anchor=north west][inner sep=0.75pt]   [align=left] {50};
% Text Node
\draw (45,41) node [anchor=north west][inner sep=0.75pt]  [color={rgb, 255:red, 255; green, 255; blue, 255 }  ,opacity=1 ] [align=left] {30};
% Text Node
\draw (22,173) node [anchor=north west][inner sep=0.75pt][color={rgb, 255:red, 255; green, 255; blue, 255 }  ,opacity=1 ] [align=left] {30};
% Text Node
\draw (143,198) node [anchor=north west][fill=white,inner sep=0.75pt] [align=left] {20};
% Text Node
\draw (98,156) node [anchor=north west][inner sep=0.75pt]   [align=left] {40};
% Text Node
\draw (201,77) node [anchor=north west][inner sep=0.75pt]  [color={rgb, 255:red, 255; green, 255; blue, 255 }  ,opacity=1 ] [align=left] {30};
% Text Node
\draw (173,169) node [anchor=north west][inner sep=0.75pt]  [color={rgb, 255:red, 255; green, 255; blue, 255 }  ,opacity=1 ] [align=left] {30};
% Text Node
\draw (110,112) node [anchor=north west][inner sep=0.75pt]   [align=left] {10};
% Text Node
\draw (229.2,132) node [anchor=north west][inner sep=0.75pt]   [align=left] {10};

\end{tikzpicture}

%% file: tikz/dessin12.tex
\tikzset{every picture/.style={line width=0.75pt}} %set default line width to 0.75pt        

\begin{tikzpicture}[x=0.75pt,y=0.75pt,yscale=-1,xscale=1]
%uncomment if require: \path (0,454); %set diagram left start at 0, and has height of 454

%Shape: Circle [id:dp9405617820437235] 
\draw  [fill={rgb, 255:red, 184; green, 233; blue, 134 }  ,fill opacity=1 ] (119,88) .. controls (119,79.16) and (126.16,72) .. (135,72) .. controls (143.84,72) and (151,79.16) .. (151,88) .. controls (151,96.84) and (143.84,104) .. (135,104) .. controls (126.16,104) and (119,96.84) .. (119,88) -- cycle ;
%Shape: Circle [id:dp9654588849964871] 
\draw  [fill={rgb, 255:red, 184; green, 233; blue, 134 }  ,fill opacity=1 ] (13,148) .. controls (13,139.16) and (20.16,132) .. (29,132) .. controls (37.84,132) and (45,139.16) .. (45,148) .. controls (45,156.84) and (37.84,164) .. (29,164) .. controls (20.16,164) and (13,156.84) .. (13,148) -- cycle ;
%Straight Lines [id:da01307563301109771] 
\draw    (47,74) -- (29,132) ;
%Shape: Circle [id:dp9242607935108363] 
\draw  [fill={rgb, 255:red, 245; green, 166; blue, 35 }  ,fill opacity=1 ] (38,59) .. controls (38,50.16) and (45.16,43) .. (54,43) .. controls (62.84,43) and (70,50.16) .. (70,59) .. controls (70,67.84) and (62.84,75) .. (54,75) .. controls (45.16,75) and (38,67.84) .. (38,59) -- cycle ;
%Straight Lines [id:da2669065491635544] 
\draw    (119,88) -- (68,65) ;
%Shape: Circle [id:dp8882081676739115] 
\draw  [fill={rgb, 255:red, 245; green, 166; blue, 35 }  ,fill opacity=1 ] (95,164) .. controls (95,155.16) and (102.16,148) .. (111,148) .. controls (119.84,148) and (127,155.16) .. (127,164) .. controls (127,172.84) and (119.84,180) .. (111,180) .. controls (102.16,180) and (95,172.84) .. (95,164) -- cycle ;
%Straight Lines [id:da23088446350661307] 
\draw    (118,149) -- (130,103) ;
%Straight Lines [id:da4868831584059441] 
\draw    (45,148) -- (95,164) ;
%Rounded Rect [id:dp00778196706745693] 
\draw  [fill={rgb, 255:red, 200; green, 222; blue, 248 }  ,fill opacity=1 ] (26,91.2) .. controls (26,88.88) and (27.88,87) .. (30.2,87) -- (50.8,87) .. controls (53.12,87) and (55,88.88) .. (55,91.2) -- (55,103.8) .. controls (55,106.12) and (53.12,108) .. (50.8,108) -- (30.2,108) .. controls (27.88,108) and (26,106.12) .. (26,103.8) -- cycle ;
%Rounded Rect [id:dp5631644558493838] 
\draw  [fill={rgb, 255:red, 200; green, 222; blue, 248 }  ,fill opacity=1 ] (54,147.2) .. controls (54,144.88) and (55.88,143) .. (58.2,143) -- (78.8,143) .. controls (81.12,143) and (83,144.88) .. (83,147.2) -- (83,159.8) .. controls (83,162.12) and (81.12,164) .. (78.8,164) -- (58.2,164) .. controls (55.88,164) and (54,162.12) .. (54,159.8) -- cycle ;
%Rounded Rect [id:dp4547345290709932] 
\draw  [fill={rgb, 255:red, 200; green, 222; blue, 248 }  ,fill opacity=1 ] (78,68.2) .. controls (78,65.88) and (79.88,64) .. (82.2,64) -- (102.8,64) .. controls (105.12,64) and (107,65.88) .. (107,68.2) -- (107,80.8) .. controls (107,83.12) and (105.12,85) .. (102.8,85) -- (82.2,85) .. controls (79.88,85) and (78,83.12) .. (78,80.8) -- cycle ;
%Rounded Rect [id:dp007241138291997462] 
\draw  [fill={rgb, 255:red, 200; green, 222; blue, 248 }  ,fill opacity=1 ] (108,115.2) .. controls (108,112.88) and (109.88,111) .. (112.2,111) -- (132.8,111) .. controls (135.12,111) and (137,112.88) .. (137,115.2) -- (137,127.8) .. controls (137,130.12) and (135.12,132) .. (132.8,132) -- (112.2,132) .. controls (109.88,132) and (108,130.12) .. (108,127.8) -- cycle ;
%Shape: Circle [id:dp773462502521179] 
\draw  [fill={rgb, 255:red, 184; green, 233; blue, 134 }  ,fill opacity=1 ] (277,89) .. controls (277,80.16) and (284.16,73) .. (293,73) .. controls (301.84,73) and (309,80.16) .. (309,89) .. controls (309,97.84) and (301.84,105) .. (293,105) .. controls (284.16,105) and (277,97.84) .. (277,89) -- cycle ;
%Shape: Circle [id:dp27036771890979516] 
\draw  [fill={rgb, 255:red, 184; green, 233; blue, 134 }  ,fill opacity=1 ] (171,149) .. controls (171,140.16) and (178.16,133) .. (187,133) .. controls (195.84,133) and (203,140.16) .. (203,149) .. controls (203,157.84) and (195.84,165) .. (187,165) .. controls (178.16,165) and (171,157.84) .. (171,149) -- cycle ;
%Straight Lines [id:da239481544789245] 
\draw    (205,75) -- (187,133) ;
%Shape: Circle [id:dp15550407315842218] 
\draw  [fill={rgb, 255:red, 245; green, 166; blue, 35 }  ,fill opacity=1 ] (196,60) .. controls (196,51.16) and (203.16,44) .. (212,44) .. controls (220.84,44) and (228,51.16) .. (228,60) .. controls (228,68.84) and (220.84,76) .. (212,76) .. controls (203.16,76) and (196,68.84) .. (196,60) -- cycle ;
%Straight Lines [id:da8493194371699205] 
\draw    (277,89) -- (226,66) ;
%Shape: Circle [id:dp4073510595354749] 
\draw  [fill={rgb, 255:red, 245; green, 166; blue, 35 }  ,fill opacity=1 ] (253,165) .. controls (253,156.16) and (260.16,149) .. (269,149) .. controls (277.84,149) and (285,156.16) .. (285,165) .. controls (285,173.84) and (277.84,181) .. (269,181) .. controls (260.16,181) and (253,173.84) .. (253,165) -- cycle ;
%Straight Lines [id:da0009219957285269142] 
\draw    (199,136) -- (281,98) ;
%Straight Lines [id:da258679910212428] 
\draw    (203,149) -- (253,165) ;
%Rounded Rect [id:dp5209502125652921] 
\draw  [fill={rgb, 255:red, 200; green, 222; blue, 248 }  ,fill opacity=1 ] (184,92.2) .. controls (184,89.88) and (185.88,88) .. (188.2,88) -- (208.8,88) .. controls (211.12,88) and (213,89.88) .. (213,92.2) -- (213,104.8) .. controls (213,107.12) and (211.12,109) .. (208.8,109) -- (188.2,109) .. controls (185.88,109) and (184,107.12) .. (184,104.8) -- cycle ;
%Rounded Rect [id:dp8777858574910671] 
\draw  [fill={rgb, 255:red, 200; green, 222; blue, 248 }  ,fill opacity=1 ] (212,148.2) .. controls (212,145.88) and (213.88,144) .. (216.2,144) -- (236.8,144) .. controls (239.12,144) and (241,145.88) .. (241,148.2) -- (241,160.8) .. controls (241,163.12) and (239.12,165) .. (236.8,165) -- (216.2,165) .. controls (213.88,165) and (212,163.12) .. (212,160.8) -- cycle ;
%Rounded Rect [id:dp48963310063874044] 
\draw  [fill={rgb, 255:red, 200; green, 222; blue, 248 }  ,fill opacity=1 ] (236,69.2) .. controls (236,66.88) and (237.88,65) .. (240.2,65) -- (260.8,65) .. controls (263.12,65) and (265,66.88) .. (265,69.2) -- (265,81.8) .. controls (265,84.12) and (263.12,86) .. (260.8,86) -- (240.2,86) .. controls (237.88,86) and (236,84.12) .. (236,81.8) -- cycle ;
%Rounded Rect [id:dp13986364042612498] 
\draw  [fill={rgb, 255:red, 200; green, 222; blue, 248 }  ,fill opacity=1 ] (227,108.2) .. controls (227,105.88) and (228.88,104) .. (231.2,104) -- (251.8,104) .. controls (254.12,104) and (256,105.88) .. (256,108.2) -- (256,120.8) .. controls (256,123.12) and (254.12,125) .. (251.8,125) -- (231.2,125) .. controls (228.88,125) and (227,123.12) .. (227,120.8) -- cycle ;
%Shape: Circle [id:dp5746330659839559] 
\draw  [fill={rgb, 255:red, 184; green, 233; blue, 134 }  ,fill opacity=1 ] (431,87) .. controls (431,78.16) and (438.16,71) .. (447,71) .. controls (455.84,71) and (463,78.16) .. (463,87) .. controls (463,95.84) and (455.84,103) .. (447,103) .. controls (438.16,103) and (431,95.84) .. (431,87) -- cycle ;
%Shape: Circle [id:dp7341387008516017] 
\draw  [fill={rgb, 255:red, 245; green, 166; blue, 35 }  ,fill opacity=1 ] (325,147) .. controls (325,138.16) and (332.16,131) .. (341,131) .. controls (349.84,131) and (357,138.16) .. (357,147) .. controls (357,155.84) and (349.84,163) .. (341,163) .. controls (332.16,163) and (325,155.84) .. (325,147) -- cycle ;
%Straight Lines [id:da6413442286399855] 
\draw    (359,73) -- (341,131) ;
%Shape: Circle [id:dp696241009072102] 
\draw  [fill={rgb, 255:red, 245; green, 166; blue, 35 }  ,fill opacity=1 ] (350,58) .. controls (350,49.16) and (357.16,42) .. (366,42) .. controls (374.84,42) and (382,49.16) .. (382,58) .. controls (382,66.84) and (374.84,74) .. (366,74) .. controls (357.16,74) and (350,66.84) .. (350,58) -- cycle ;
%Straight Lines [id:da31586388900882467] 
\draw    (431,87) -- (380,64) ;
%Shape: Circle [id:dp20699783420589246] 
\draw  [fill={rgb, 255:red, 184; green, 233; blue, 134 }  ,fill opacity=1 ] (407,163) .. controls (407,154.16) and (414.16,147) .. (423,147) .. controls (431.84,147) and (439,154.16) .. (439,163) .. controls (439,171.84) and (431.84,179) .. (423,179) .. controls (414.16,179) and (407,171.84) .. (407,163) -- cycle ;
%Straight Lines [id:da11575347070190234] 
\draw    (353,134) -- (435,96) ;
%Straight Lines [id:da7231399553513524] 
\draw    (357,147) -- (407,163) ;
%Rounded Rect [id:dp10417235656181856] 
\draw  [fill={rgb, 255:red, 200; green, 222; blue, 248 }  ,fill opacity=1 ] (338,90.2) .. controls (338,87.88) and (339.88,86) .. (342.2,86) -- (362.8,86) .. controls (365.12,86) and (367,87.88) .. (367,90.2) -- (367,102.8) .. controls (367,105.12) and (365.12,107) .. (362.8,107) -- (342.2,107) .. controls (339.88,107) and (338,105.12) .. (338,102.8) -- cycle ;
%Rounded Rect [id:dp7560213459248769] 
\draw  [fill={rgb, 255:red, 200; green, 222; blue, 248 }  ,fill opacity=1 ] (366,146.2) .. controls (366,143.88) and (367.88,142) .. (370.2,142) -- (390.8,142) .. controls (393.12,142) and (395,143.88) .. (395,146.2) -- (395,158.8) .. controls (395,161.12) and (393.12,163) .. (390.8,163) -- (370.2,163) .. controls (367.88,163) and (366,161.12) .. (366,158.8) -- cycle ;
%Rounded Rect [id:dp5848156013409594] 
\draw  [fill={rgb, 255:red, 200; green, 222; blue, 248 }  ,fill opacity=1 ] (390,67.2) .. controls (390,64.88) and (391.88,63) .. (394.2,63) -- (414.8,63) .. controls (417.12,63) and (419,64.88) .. (419,67.2) -- (419,79.8) .. controls (419,82.12) and (417.12,84) .. (414.8,84) -- (394.2,84) .. controls (391.88,84) and (390,82.12) .. (390,79.8) -- cycle ;
%Rounded Rect [id:dp6803360186594051] 
\draw  [fill={rgb, 255:red, 200; green, 222; blue, 248 }  ,fill opacity=1 ] (381,106.2) .. controls (381,103.88) and (382.88,102) .. (385.2,102) -- (405.8,102) .. controls (408.12,102) and (410,103.88) .. (410,106.2) -- (410,118.8) .. controls (410,121.12) and (408.12,123) .. (405.8,123) -- (385.2,123) .. controls (382.88,123) and (381,121.12) .. (381,118.8) -- cycle ;

% Text Node
\draw (23,139) node [anchor=north west][inner sep=0.75pt]   [align=left] {A};
% Text Node
\draw (48,50) node [anchor=north west][inner sep=0.75pt]   [align=left] {B};
% Text Node
\draw (128,79) node [anchor=north west][inner sep=0.75pt]   [align=left] {C};
% Text Node
\draw (105,155) node [anchor=north west][inner sep=0.75pt]   [align=left] {D};
% Text Node
\draw (31.2,89) node [anchor=north west][inner sep=0.75pt]   [align=left] {50};
% Text Node
\draw (59.2,145) node [anchor=north west][inner sep=0.75pt]   [align=left] {50};
% Text Node
\draw (83.2,66) node [anchor=north west][inner sep=0.75pt]   [align=left] {50};
% Text Node
\draw (113.2,113) node [anchor=north west][inner sep=0.75pt]   [align=left] {50};
% Text Node
\draw (124.2,40) node [anchor=north west][inner sep=0.75pt]  [color={rgb, 255:red, 255; green, 255; blue, 255 }  ,opacity=1 ] [align=left] {50};
% Text Node
\draw (80,10) node [anchor=north west][inner sep=0.75pt]   [align=left] {(a)};
% Text Node
\draw (181,140) node [anchor=north west][inner sep=0.75pt]   [align=left] {A};
% Text Node
\draw (206,51) node [anchor=north west][inner sep=0.75pt]   [align=left] {B};
% Text Node
\draw (286,80) node [anchor=north west][inner sep=0.75pt]   [align=left] {C};
% Text Node
\draw (263,156) node [anchor=north west][inner sep=0.75pt]   [align=left] {D};
% Text Node
\draw (189.2,90) node [anchor=north west][inner sep=0.75pt]   [align=left] {50};
% Text Node
\draw (217.2,146) node [anchor=north west][inner sep=0.75pt]   [align=left] {50};
% Text Node
\draw (241.2,67) node [anchor=north west][inner sep=0.75pt]   [align=left] {50};
% Text Node
\draw (232.2,106) node [anchor=north west][inner sep=0.75pt]   [align=left] {50};
% Text Node
\draw (282.2,41) node [anchor=north west][inner sep=0.75pt]  [color={rgb, 255:red, 255; green, 255; blue, 255 }  ,opacity=1 ] [align=left] {50};
% Text Node
\draw (238,11) node [anchor=north west][inner sep=0.75pt]   [align=left] {(b)};
% Text Node
\draw (335,138) node [anchor=north west][inner sep=0.75pt]   [align=left] {A};
% Text Node
\draw (360,49) node [anchor=north west][inner sep=0.75pt]   [align=left] {B};
% Text Node
\draw (440,78) node [anchor=north west][inner sep=0.75pt]   [align=left] {C};
% Text Node
\draw (417,154) node [anchor=north west][inner sep=0.75pt]   [align=left] {D};
% Text Node
\draw (343.2,88) node [anchor=north west][inner sep=0.75pt]   [align=left] {50};
% Text Node
\draw (371.2,144) node [anchor=north west][inner sep=0.75pt]   [align=left] {50};
% Text Node
\draw (395.2,65) node [anchor=north west][inner sep=0.75pt]   [align=left] {50};
% Text Node
\draw (386.2,104) node [anchor=north west][inner sep=0.75pt]   [align=left] {50};
% Text Node
\draw (436.2,39) node [anchor=north west][inner sep=0.75pt]  [color={rgb, 255:red, 255; green, 255; blue, 255 }  ,opacity=1 ] [align=left] {50};
% Text Node
\draw (392,9) node [anchor=north west][inner sep=0.75pt]   [align=left] {(c)};

\end{tikzpicture}